\documentclass[11pt,a4paper]{article}

\usepackage{amsmath,amssymb,amsthm}

\usepackage{fullpage}
\usepackage{cite}
\usepackage{chngcntr}

\DeclareMathOperator{\tr}{tr}

\DeclareMathOperator*{\essinf}{ess\,inf}
\DeclareMathOperator*\sgn{sgn}
\newcommand{\Et}[1]{E\left[\left.#1\,\right|\mathcal F_t\right]}

\numberwithin{equation}{section}

\theoremstyle{plain}
\newtheorem{theorem}{Theorem}[section]
\newtheorem{proposition}[theorem]{Proposition}
\newtheorem{corollary}[theorem]{Corollary}
\newtheorem{lemma}[theorem]{Lemma}
\theoremstyle{definition}
\newtheorem{definition}[theorem]{Definition}
\newtheorem{assumption}[theorem]{Assumption}

\theoremstyle{remark}
\newtheorem{remark}[theorem]{Remark}
\newtheorem{example}[theorem]{Example}
\usepackage{parskip}
\makeatletter
\def\thm@space@setup{%
  \thm@preskip=\parskip \thm@postskip=0pt
}
\makeatother

\title{Smooth Solutions to Portfolio Liquidation Problems\\ under Price-Sensitive Market Impact\thanks{We thank seminar participants at various institutions for valuable comments and suggestions. Financial support through the \textit{CRC 649 Economic Risk} and \textit{d-fine GmbH} is gratefully acknowledged. Parts of this manuscript were written while Horst was visiting the Hausdorff Research Institute for Mathematics in Bonn and the Center for Interdisciplinary Research at Bielefeld University; grateful acknowledgment is made for hospitality. We are indebted to two anonymous referees for their careful reading of the manuscript and their many suggestions which greatly improved the quality of this manuscript.}}

\author{
Paulwin Graewe\footnote{Department of Mathematics, Humboldt-Universit\"at zu Berlin, Unter den Linden 6, D-10099 Berlin, Germany, \texttt{graewe@math.hu-berlin.de}} \and
Ulrich Horst\footnote{Department of Mathematics and School of Business and Economics, Humboldt-Universit\"at zu Berlin, Unter den Linden 6, D-10099 Berlin, Germany, \texttt{horst@math.hu-berlin.de}} \and
Eric S\'{e}r\'{e}\footnote{Universit\'{e} Paris-Dauphine, PSL Research University, CNRS, UMR 7534, CEREMADE, 75016 Paris, France, \texttt{sere@ceremade.dauphine.fr}} 
}

\begin{document}

\maketitle

\begin{abstract}
We consider the stochastic control problem of a financial trader that needs to unwind a large asset portfolio within a short period of time. The trader can simultaneously submit active orders to a primary market and passive orders to a dark pool. Our framework is flexible enough to allow for price-dependent impact functions describing the trading costs in the primary market and price-dependent adverse selection costs associated with dark pool trading. We prove that the value function can be characterized in terms of the unique smooth solution to a PDE with singular terminal value, establish its explicit asymptotic behavior at the terminal time, and give the optimal trading strategy in feedback form.  
\end{abstract}



\section{Introduction}
	Traditional financial market models assume that asset prices follow an exogenous stochastic process and that all transactions can be settled without any impact on market prices. This assumption is appropriate for small investors who trade only a negligible proportion of the average daily trading volume. It is not always appropriate, though, for institutional investors trading large blocks of shares over a short time span. 

	The analysis of optimal liquidation problems has received considerable attention in the mathematical finance and stochastic control literature in recent years. Starting with the paper of  Almgren \& Chriss~\cite{AlmgrenChriss00} existence and uniqueness results of optimal liquidation strategies under various market regimes and price impact functions have been established by many authors, including \cite{AnkirchnerJeanblancKruse14,AnkirchnerKruse12,ChenKouWang15,Forsyth2012,GatheralSchied11,GraeweHorstQiu13,HorstNaujokat14,KloeckSchiedSun13,Kratz14,KratzSchoeneborn15,Schied13,SchiedSchoenebornTehranchi10}. One of the main characteristics of stochastic optimization problems arising in portfolio liquidation models is the singular terminal condition of the value function induced by the liquidation constraint. The singularity is already present, yet not immediately visible, in the original price impact model of Almgren \& Chriss~\cite{AlmgrenChriss00}. Within their mean variance framework and with arithmetic Brownian motion as the benchmark price process, the objective function is deterministic, and the optimization problem is essentially a classical variational problem where the terminal state constraint causes no further difficulties. However, when considering a geometric Brownian motion as the underlying price process as in Forsyth~et~al.~\cite{Forsyth2012}, the optimal execution strategies become price-sensitive. One is then faced with a genuine stochastic control problem where the singularity becomes a challenge when determining the value function and applying verification arguments. 

	Several approaches to overcome this challenge have recently been suggested in the stochastic control literature. Forsyth et al.~\cite{Forsyth2012} solve the control problem numerically by penalizing open positions at the final time.
	Ankirchner \& Kruse~\cite{AnkirchnerKruse12} characterize the value function of a Markovian liquidation problem as the unique viscosity solution to the Hamilton-Jacobi-Bellman (HJB) equation. Their verification argument uses a discrete approximation of the continuous time model. Ankirchner et al.~\cite{AnkirchnerJeanblancKruse14}, Graewe et al.~\cite{GraeweHorstQiu13}, and Horst et al.~\cite{HorstQiuZhang16} consider non-Markovian liquidation problems where the cost functional is driven by general adapted factor processes and the HJB equation solves a BSDE or BSPDE, depending on the dynamics of the factor processes. In all three cases, existence of solutions to the HJB equation is established by analyzing limits of sequences of BS(P)DEs with increasing finite terminal values while the verification argument uses generalized It\^{o}-Kunita formulas \cite{GraeweHorstQiu13}, resp., the link between degenerate BSPDEs and forward-backward stochastic differential equations \cite{HorstQiuZhang16}.

A general class of Markovian liquidation problems has been solved in Schied~\cite{Schied13} by means of Dawson--Watanabe superprocess. This approach avoids the use of HJB equations. Instead, it uses a probabilistic verification argument based on log-Laplace functionals of superprocesses that requires sharp upper and lower bounds for the candidate value function. 

	This paper establishes existence of a smooth solution to a class of Markovian portfolio liquidation problems. While existence of a weakly differentiable solution to our HJB equation can be inferred from \cite{GraeweHorstQiu13} and existence of optimal liquidation strategies can be inferred from, e.g., \cite{AnkirchnerJeanblancKruse14,GraeweHorstQiu13}, smooth solutions to stochastic portfolio liquidation problems have not yet been established in the literature before. As in~\cite{Schied13} the key is to know the precise asymptotic behavior of the value function at the terminal time. The asymptotics allows us to characterize the HJB equation in terms of a PDE with \textit{finite} terminal value yet singular nonlinearity, for which existence of a unique smooth solution can be proved using standard fixed point arguments in a suitable function space. 
	 
	As in \cite{GraeweHorstQiu13, HorstNaujokat14, KratzSchoeneborn15} we allow for simultaneous submission of active orders, i.e., orders for immediate execution at the best available price, to a primary venue, and of passive orders, i.e., orders for future execution, into a dark pool. Dark pools are alternative trading venues that allow investors to reduce market impact and hence trading costs by submitting liquidity that is shielded from public view. Trade execution is uncertain, though, as trades will be settled only if matching liquidity becomes available. 

	Active orders incur market impact costs while passive orders incur adverse selection (or ``slippage'') costs. In our model impact, adverse selection and risk costs are driven by a multi-dimensional time-homogeneous diffusion process. We think of the multi-dimensional diffusion process as describing the joint dynamics of a fundamental or benchmark stock price process and of a stochastic factor process driving the fluctuations in the available liquidity over time. Assuming that all cost terms are of the same order $p>1$ allows us to make the standard separation ansatz of expressing the value function as $v(t,y)|x|^p$ where $y$ is the state of the factor process and $x$ is the portfolio position. We establish sharp a priori estimates on $v$ at the terminal time. This is similar to~\cite{Schied13} and avoids the stochastic penalization method applied in \cite{AnkirchnerJeanblancKruse14, GraeweHorstQiu13, HorstQiuZhang16}. 
	
	For deterministic cost coefficients the asymptotics of $v$ can be inferred from, e.g., Seidman \& Yong~\cite{SeidmanYong96} who give the asymptotics of minimal $L^p$-norm controls of multi-dimensional linear systems with terminal state constraint. To the best of our knowledge a corresponding result for stochastic systems is not available in the literature.\footnote{The literature on minimal energy problems in stochastic settings (see \cite{Klamka13} and the references therein) seems to focus on controllability which is trivial in the present one-dimensional setting.} 
	
	Our a priori estimates allow us to prove the uniqueness of a continuous viscosity solution of polynomial growth to the HJB equation under continuity and polynomial growth conditions on the cost coefficients. The proof uses a comparison result for parabolic PDEs with possibly infinite terminal values. As a byproduct we obtain that the minimal nonnegative solution to the stochastic HJB equation in \cite{AnkirchnerJeanblancKruse14} is indeed the unique nonnegative solution to their singular BSDE if the coefficients are essentially bounded. 
	
	Our main contribution is the proof of the existence of a classical solution to the HJB equation under additional smoothness and boundedness conditions on the cost coefficients. The proof is based on the a priori estimates. They provide us with the precise asymptotic behavior of the solution at the terminal time and thus allow us to reduce the problem of solving the PDE with singular terminal value for $v$ to solving a PDE with finite terminal value, yet degenerate non-linearity. Using Krylov's generalized It\^o formula we prove that the classical solution is indeed the value function and that the optimal trading strategies can be given in closed form.  
		
	The remainder of this paper is organized as follows. The stochastic control problem is formulated in Section \ref{section-control}. The a priori estimates and the comparison principle that yields uniqueness of a continuous viscosity solution is established in Section 3. Existence and uniqueness of a classical solution to the HJB equation is proven in Section 4. The verification argument is carried out in Section 5. Finally, we show in Section 6 how our uniqueness result extends to the non-Markovian case analyzed in \cite{AnkirchnerJeanblancKruse14}. 

	\textit{Notational conventions.} We denote by $C(\mathbb R^d)$ the space of \textit{bounded} continuous functions. The functions in $u \in C_{poly}([0,T]\times\mathbb R^d)$ are continuous and for some $C>0$ and $n\geq 1$,
\begin{equation} \label{polynomial-growth-condition}
|u(t,y)|\leq C(1+|y|^n), \qquad (t,y)\in [0,T]\times\mathbb R^d.
\end{equation} 
The space $C^{1,2}_{loc}([0,T]\times\mathbb R^d)$ denotes the class of the functions $u(t,y)$ that are continuous (possibly unbounded) along with their first derivative in $t$ and their first and second derivative in~$y$. 
For generic $\alpha\in{(0,1)}$ and normed vector space~$E$, the functions in $C^{k+\alpha}([0,T];E)$ have $\alpha$-H\"older continuous derivatives up to the order $k$. By $W^{2}_{q,loc}(\mathbb{R}^d)$ we denote the usual Sobolev spaces of all functions that are locally $L^q$-integrable along with their weak first and second order derivative \cite[Definition~1.62]{AdamsFournier03}. The parabolic version $W^{1,2}_{q,loc}((0,T)\times \mathbb{R}^d)$ is the space of all functions $u(t,y)$ that are locally $L^q$-integrable along with their weak first derivative in $t$ and their weak first and second derivative in $y$. Whenever the notation $T^-$ appears in the definition of a function space we mean the set of all functions whose restrictions satisfy the respective property when $T^-$ is replaced by any $s<T$, e.g., 
\[
	C_{poly}([0,T^-]\times\mathbb R^d)=\{u:[0,T)\times\mathbb R^d\rightarrow \mathbb R: u_{|[0,s]\times\mathbb R^d}\in C_{poly}([0,s]\times\mathbb R^d) \text{ for all } s\in[0,T)\}
\]
and
\[
W^{1,2}_{q,loc}((0,T^-)\times\mathbb R^d)=\{u:(0,T)\times\mathbb R^d\rightarrow \mathbb R: u_{|(0,s)\times\mathbb R^d}\in W^{1,2}_{q,loc}((0,s)\times\mathbb R^d) \text{ for all } s\in(0,T)\}.
\]
The set of adapted $\mathbb R^d$-valued processes $(Z_t)_{t\in[0,T]}$ satisfying $E[\int_0^T|Z_t|^q\,dt]<\infty$ is denoted by $L^q_\mathcal F(0,T;\mathbb R^d)$; the subset of processes with continuous paths satisfying $E[\sup_{t\in[0,T]}|Z_t|^q]<\infty$ is denoted by $L^q_{\mathcal F}(\Omega;C([0,T];\mathbb R^d))$. If not otherwise indicated then $\|\cdot\|$ denotes the supremum norm. For arbitrary $\beta>0$ we occasionally write $\sqrt[\beta]{\,\cdot\,}$ instead of $(\,\cdot\,)^{1/\beta}$. All equations are to be understood in the a.s.\ sense.



\section{Problem formulation, assumptions and main results}\label{section-control}
	We consider the stochastic optimization problem of an investor that needs to close a (large) position of  shares within a given time interval $[0,T]$. Following Horst \& Naujokat~\cite{HorstNaujokat14} and Kratz \& Sch\"oneborn~\cite{KratzSchoeneborn15} the investor may trade in an absolutely continuous manner in a primary exchange and simultaneously place passive block orders into a dark pool. Execution of passive orders is modeled by the jump times of a Poisson process $(N_t)_{t\in[0,T]}$ with constant intensity $\theta\geq0$. 
	
	The Poisson process $N$ and an $n$-dimensional standard Brownian motion $(W_t)_{t\in[0,T]}$ are defined on a stochastic basis $(\Omega,\mathcal F,(\mathcal F_t)_{t\in[0,T]},\mathbb P)$ satisfying the usual conditions. In what follows we repeatedly use the independence of $N$ and $W$.
	
	As the factor process driving trading costs we consider the $d$-dimensional It\^o diffusion  
\begin{equation} \label{dynamic-price}
 Y_s^{t,y}=y+\int_t^sb(Y_r^{t,y})\,dr+\int_t^s\sigma(Y_r^{t,y})\,dW_r, \qquad t\leq s\leq T.
\end{equation}

\begin{assumption}\label{A1}
We assume throughout that the coefficients $b:\mathbb R^d\rightarrow \mathbb R^d$ and $\sigma:\mathbb R^d\rightarrow\mathbb R^{d\times n}$ are Lipschitz continuous.
\end{assumption}

The preceding assumption guarantees that \eqref{dynamic-price} admits a unique strong solution $(Y_s^{t,y})_{s\in[t,T]}$, for every initial state $(t,y)\in[0,T]\times\mathbb R^d$. Furthermore, see \cite[Theorem~3.35]{PardouxRascanu14}, the map $(s,t,y)\mapsto Y_s^{t,y}$ is a.s.\ continuous and for every $n\geq2$ there exists $C_n>0$ such that 
the following moment estimate holds with the convention $Y_s^{t,y}=y$ for $0\leq s\leq t$:
\begin{equation} \label{moment-estimate}
E[\sup\nolimits_{s\in[0,T]}|Y_s^{t,y}|^n]\leq C_n(1+|y|^n), \qquad (t,y)\in[0,T]\times\mathbb R^d.
\end{equation}
This moment estimate in particular guarantees by Vitali's convergence theorem that functions of the form 
\begin{equation} \label{Vitali}
	(t,y)\mapsto E\left[\int_t^Tf(s,Y_s^{t,y})\,ds\right],
\end{equation}
with $f\in C_{poly}([0,T]\times\mathbb R^d)$, belong again to $C_{poly}([0,T]\times\mathbb R^d)$.

\subsection{The stochastic control problem}
	For any initial time $t\in[0,T)$ and initial position $x\in\mathbb{R}$, we denote by $\mathcal A(t,x)$ the set of all admissible liquidation strategies $(\xi,\pi)$. Here, $\xi=(\xi_s)_{s\in[t,T]}$ describes the rates at which the agent trades in the primary market, while $\pi=(\pi_s)_{s\in[t,T]}$ describes the passive orders submitted to the dark pool. A pair of strategies $(\xi,\pi)$ is admissible if $\xi$ is progressively measurable and $\pi$ is predictable such that the resulting portfolio process 
\begin{equation*} 
X_s^{\xi,\pi}=x-\int_t^s\xi_r\,dr-\int_t^s\pi_r\,dN_r,  \qquad t\leq s\leq T,
\end{equation*}
satisfies the liquidation constraint 
\begin{equation}   \label{liquidation-constraint}
	X_{T}^{\xi,\pi}=0.
\end{equation}
The costs associated with an admissible liquidation strategy $(\xi,\pi)$ are modeled by the cost functional
\begin{equation*}
	J(t,y,x;\xi,\pi):=E\left[\int_t^T\eta(Y_s^{t,y}) |\xi_s|^p+\theta\gamma(Y_s^{t,y})|\pi_s|^p+\lambda(Y_s^{t,y})|X_s^{\xi,\pi}|^p\,ds\right]. 
\end{equation*}
The first term of the nonnegative running costs 
\begin{equation*}
c(y,x,\xi,\pi):=\eta(y)|\xi|^p+\theta\gamma(y)|\pi|^p+\lambda(y)|x|^p
\end{equation*}
describes the temporary price impact at the primary exchange; the second term describes adverse selection costs associated with dark pool trading (see \cite{HorstNaujokat14, KloeckSchiedSun13} for details) while the third term penalizes slow liquidation. The latter may be interpreted as the $p$-th power of the Value-at-Risk of the open position (see \cite{AnkirchnerKruse12, GatheralSchied11} for details). 

The value function of the stochastic control problem is defined for each initial state $(t,y,x)\in[0,T)\times\mathbb R^{d}\times\mathbb R$ as
\begin{equation} \label{value-function}
	V(t,y,x):=\inf_{(\xi,\pi)\in\mathcal A(t,x)}J(t,y,x;\xi,\pi).
\end{equation}

\begin{assumption}\label{A2}
We assume throughout that $p>1$\footnote{Unlike \cite{Schied13} we do not exclude exponents $1<p<2$, which correspond to root shaped temporary price impact. Almgren et al.\ \cite{Almgrenetal05} give empirical evidence for $p=8/5$.} and put $\beta:=1/(p-1)>0$. We further assume that the cost coefficients satisfy the following conditions: 
\begin{itemize}
	\item[(i)] The coefficients $\eta,\gamma,\lambda,1/\eta:\mathbb R^d\rightarrow \mathbb [0,\infty)$ are continuous.  
	\item[(ii)] The coefficients $\eta$, $\lambda, 1/\eta$ are of polynomial growth, i.e., for some $n\geq 1$ and $C>0$, 
\begin{equation} \label{polynomial-growth-coefficients}
	\eta(y)+\lambda(y)+1/\eta(y)\leq C(1+|y|^n), \qquad y\in \mathbb R^d.
\end{equation}
\end{itemize}
\end{assumption}

\begin{example}	
Our assumptions on the factor process allow us to capture simultaneously several key determinants of trading costs. The assumptions are satisfied for the arithmetic Brownian motion model $$dY^1_t = \mu\,dt + \sigma\, dW^1_t$$ as well as for a mean-reverting process of the form $$dY^2_t = f(\nu - Y^2_t)\,dt + dW^2_t$$ for a bounded Lipschitz continuous function $f$. The (logarithmic) price process $Y^1$ may drive the market risk factor $\lambda$ while $Y^2$ may describe stochastic order book heights (stochastic liquidity) and hence drive $\eta$. 
\end{example}

\begin{assumption}
In order to state our main result (Theorem \ref{thm-existence-GHS}), we will need the following additional assumptions: 
\begin{itemize}
	\item[(A1)] $\sigma\sigma^\ast$ is uniformly positive definite.
	\item[(A2)] $b$ and $\sigma$ are bounded.
	\item[(A3)] $\eta$, $1/\eta$, and $\lambda$ are bounded. In particular, $\eta\geq \kappa_0$ for some constant $\kappa_0>0$.
	\item[(A4)] $\eta$ is twice continuously differentiable, and $\mathcal L\eta$ is bounded. 
\end{itemize}
\end{assumption}

\par\medskip

\begin{remark}
Condition (A1) and (A2) provide, in particular, the fact that
\begin{equation*} 
	\left\{
	\begin{aligned}
	& D(\mathcal L)=\{u\in\bigcap\nolimits_{q\geq 1}W_{q,loc}^{2}(\mathbb R^d):u, \mathcal L u\in C(\mathbb R^d)\}\\
	&\mathcal L:D(\mathcal L)\subset C(\mathbb R^d)\rightarrow C(\mathbb R^d)
	\end{aligned}
	\right.
\end{equation*}
is a sectorial realization of the operator $\mathcal L$ in $C(\mathbb R^d)$ and hence that $\mathcal L$ generates an analytic semigroup in $C(\mathbb R^d)$, see \cite[Corollary~3.1.9]{Lunardi95}. If $d=1$, then $D(\mathcal L) = C^2(\mathbb{R})$. We use this fact in order to prove the existence of classical solutions. At the expense of additional work, it might be possible to relax (A1) and (A2), e.g., by the assumption that $b$ and $\sigma$ have a continuous and bounded second derivative (cf.~Remark~\ref{remark-pi-semigroup}), to incorporate the Ornstein-Uhlenbeck process as well as the geometric Brownian motion. We leave this as an open problem.
\end{remark}

\subsection{Heuristics and the main result}
The dynamic programing principle suggests that the value function satisfies the following HJB equation, cf.~\cite[Theorem~2.2]{ShiWu11}: 
\begin{equation} \label{hjb}
-\partial_t V(t,y,x)-\mathcal L V(t,y,x)-\inf_{\xi,\pi\in\mathbb R}	H(t,y,x,\xi,\pi,V)=0, \qquad (t,y,x)\in[0,T)\times\mathbb R^d\times\mathbb R,
\end{equation}
where
\[ 
	\mathcal L:=\frac{1}{2}\tr(\sigma\sigma^* D_y^2)+\left\langle b,D_y\right\rangle 
\]
denotes the infinitesimal generator of the factor process, and the Hamiltonian $H$ is given by
\[
	H(t,y,x,\xi,\pi,V):=-\xi \partial_xV(t,y,x)+\theta(V(t,y,x-\pi)-V(t,y,x))+c(y,x,\xi,\pi).
\]

The specific structure of our control problem with respect to the state variable~$x$ -- linear in the control dynamics and of $p$-th power in the running costs -- suggests an ansatz of the form:
\begin{equation} \label{value-function-structure}
 V(t,y,x)=v(t,y)|x|^p.
\end{equation}
Recalling that $\beta=1/(p-1)$, the proof of Lemma \ref{lemma-hjb} below is in fact standard. Before stating the lemma we formulate the different solution concepts for parabolic equations used in this paper. 

\begin{definition}
	For continuous functions $v:[0,T)\times\mathbb R^d\rightarrow \mathbb R$ we use the following solution concepts to parabolic PDEs
\begin{equation} \label{general-pde}
-\partial_t v(t,y)-H(t,y,v(t,y),D_yv(t,y),D^2_yv(t,y))=0,
\end{equation}
where $H:[0,T)\times\mathbb R^d\times\mathbb R\times\mathbb R^d\times\mathbb S^d\rightarrow \mathbb R$ and $\mathbb S^d$ denotes the set of symmetric $d\times d$ matrices.
\begin{itemize}
	\item[(i)] $v$ is a \textit{classical solution} if $v\in C^{1,2}_{loc}([0,T)\times\mathbb R^d)$ such that \eqref{general-pde} is satisfied for all $(t,y)\in[0,T)\times\mathbb R^d$.
	\item[(ii)] $v$ is a \textit{strong solution} if $v\in W^{1,2}_{1,loc}((0,T)\times\mathbb R^d)$ such that \eqref{general-pde} is satisfied in terms of the weak derivatives of $v$ a.e.\ in $[0,T)\times\mathbb R^d$.
	\item[(iii)] $v$ is a \textit{viscosity subsolution} if for every $\varphi\in C^{1,2}_{loc}([0,T)\times \mathbb R^d)$ such that $\varphi\geq v$ and $\varphi(t,y)=v(t,y)$ at a point $(t,y)\in[0,T)\times\mathbb R^d$ it holds \[-\partial_t \varphi(t,y)-H(t,y,v(t,y),D_y\varphi(t,y),D^2_y\varphi(t,y))\leq0.\]
	\item[(iv)] $v$ is a \textit{viscosity supersolution} if for every $\varphi\in C^{1,2}_{loc}([0,T)\times \mathbb R^d)$ such that $\varphi\leq v$ and $\varphi(t,y)=v(t,y)$ at a point $(t,y)\in[0,T)\times\mathbb R^d$ it holds \[-\partial_t \varphi(t,y)-H(t,y,v(t,y),D_y\varphi(t,y),D^2_y\varphi(t,y))\geq0.\]
	\item[(v)] $v$ is a \textit{viscosity solution} if $v$ is both viscosity sub- and supersolution.
\end{itemize}
\end{definition}

We recall the well-known fact that classical solutions of parabolic PDEs are also viscosity solutions, see, e.g., \cite{CrandallIshiiLions92}.

\begin{lemma} \label{lemma-hjb}
	A nonnegative function $v:[0,T)\times\mathbb R^d\rightarrow[0,\infty)$ is a classical/strong/viscosity (sub-/super-)solution to
\begin{equation} \label{inflator-pde}
-\partial_tv(t,y)-\mathcal Lv(t,y)-F(y,v(t,y))=0,
\end{equation} 
where 
\begin{equation} \label{nonlinearity}
	F(y,v):= \lambda(y)-\frac{|v|^{\beta+1}}{\beta\eta(y)^\beta}+\frac{\theta\gamma(y)v}{\sqrt[\beta]{\gamma(y)^\beta+|v|^\beta}}-\theta v,
\end{equation}
if and only if $v(t,y)|x|^p$ is a  classical/strong/viscosity (sub-/super-)solution to the HJB equation~\eqref{hjb}.  In this case the infimum in~\eqref{hjb} is attained at 
\begin{equation} \label{optimal-control-hjb} 
\xi^*(t,y,x)=\frac{v(t,y)^\beta}{\eta(y)^\beta}x \quad \text{ and } \quad \pi^*(t,y,x)=\frac{v(t,y)^\beta}{\gamma(y)^\beta+v(t,y)^\beta}x
\end{equation}
and
\begin{equation} \label{optimal-hamiltonian}
H(t,y,x,\xi^*(t,y,x),\pi^*(t,y,x),v(\cdot,\cdot)|\cdot|^p)=F(y,v(t,y))|x|^p.
\end{equation}
\end{lemma}
\begin{proof}
	When testing the viscosity property of $V=v(\cdot,\cdot)|\cdot|^p$ at a point $(\bar t, \bar y,\bar x)\in [0,T)\times\mathbb R^d\times\mathbb R$ it is sufficient to consider test functions of the form 
\begin{equation} \label{test-function}	
	\varphi = \phi(\cdot,\cdot)|\cdot|^p. 
\end{equation}	
In fact, let $\varphi\leq V$ (the case $\varphi\geq V$ is similar) be an \textit{arbitrary} test function such that $\varphi(\bar t,\bar y,\bar x)= V(\bar t,\bar y,\bar x)$. 	
	If $\bar x\neq 0$, we may define  
\[
	\tilde\varphi(t,y,x)=\frac{\varphi(t,y,\bar x)}{|\bar x|^p}|x|^p.
\]
Then, $\tilde\varphi \leq V$,  $\tilde\varphi(\bar t,\bar y,\bar x)=V(\bar t,\bar y,\bar x)$, $\partial_t\tilde\varphi(\bar t,\bar y,\bar x)=\partial_t\varphi(\bar t,\bar y,\bar x)$, $\mathcal  L\tilde\varphi(\bar t,\bar y,\bar x)=\mathcal L\varphi(\bar t,\bar y,\bar x)$. Since $V$ is continuously differentiable in $x$ and $(\bar t,\bar y,\bar x)$ is a extreme point of both $V-\tilde\varphi$ and $V-\varphi$ we furthermore have that $\partial_x\tilde\varphi(\bar t,\bar y,\bar x)=\partial_xV(\bar t,\bar y,\bar x)=\partial_x\varphi(\bar t,\bar y,\bar x)$. Testing at a point with $\bar x=0$ is trivial because $V(t,y,0)\equiv0$ and the derivatives $\partial_tV(t,y,0)$ and $\mathcal LV(t,y,0)$ exist and are identically equal to zero. As a result, we may w.l.o.g.~restrict ourselves to test functions of the form \eqref{test-function}. With this observation at hand the statement is then verified by straightforward calculations using that for $V=v(\cdot,\cdot)|\cdot|^p$ with $v\geq0$ the Hamiltonian $H$ is convex in $\xi$ and~$\pi$.  
\end{proof}

To guarantee the uniqueness of a viscosity solution to \eqref{inflator-pde} we need to impose a suitable terminal condition. Due to the liquidation constraint~\eqref{liquidation-constraint}, we expect the value function to tend to infinity for any fixed non-trivial portfolio position as $t\rightarrow T$. More precisely, when disregarding any adverse selection and risk costs, as well as any scenarios in which passive orders are executed, and using $\mathbb P(\text{no jumps of $N$ in $[t,T]$})=e^{-\theta(T-t)}$, one obtains for any admissible control $(\xi,\pi)\in\mathcal A(t,x)$,
\begin{equation} \label{no-jumps}
J(t,y,x;\xi,\pi)\geq e^{-\theta(T-t)} E\left[\left.\int_t^T\eta(Y_s^{t,y})|\xi_s|^p\,ds\right|\text{no jumps of $N$ in } [t,T]\right].
\end{equation}
Applying the reverse H\"older inequality to the inner integral of the RHS of \eqref{no-jumps} yields, 
\begin{align*}
	J(t,y,x;\xi,\pi)&\geq e^{-\theta(T-t)}E\left[\left.\left(\int_t^T\eta(Y_s^{t,y})^{\frac{-1}{p-1}}\,ds\right)^{\!-(p-1)}\left(\int_t^T|\xi_s|\,ds\right)^p\right|\text{no jumps of $N$ in } [t,T]\right].
\end{align*}
Given that no jumps occur, the liquidation constraint~\eqref{liquidation-constraint} yields $\int_t^T \xi_s = x$. Using the independence of $Y^{t,y}$ and $N$ we hence obtain, 
\begin{align*}
	J(t,y,x;\xi,\pi)	&\geq  e^{-\theta(T-t)}E\left[\left.\left(\int_t^T\eta(Y_s^{t,y})^{\frac{-1}{p-1}}\,ds\right)^{-(p-1)}\left|\int_t^T\xi_s\,ds\right|^p\right|\text{no jumps of $N$ in } [t,T]\right]\\
	&= E\left[\frac{e^{-\theta(T-t)}}{\sqrt[\beta]{\int_t^T\frac{1}{\eta(Y_s^{t,y})^\beta}\,ds}}\right]|x|^p.
\end{align*}
In view of \eqref{polynomial-growth-coefficients} and \eqref{moment-estimate}, we therefore expect that
\begin{equation*} 
\lim_{t\rightarrow T}v(t,y)= +\infty \quad \text{ locally uniformly on $\mathbb R^d$.}
\end{equation*}
It turns out that this singular terminal condition along with the standing assumptions on the diffusion and cost coefficients already ensures uniqueness of a viscosity solution to the HJB equation. The deterministic closing strategy that liquidates at a constant rate and uses no dark pool incurs the cost 
\begin{equation} \label{upper-bound}
\begin{split}
	J\Big(t,y,x;\frac{x}{T-t},0 \Big) & = \frac{1}{(T-t)^p} E\left[\int_t^T\eta(Y_s^{t,y}) + (T-s)^p \lambda(Y_s^{t,y}) \,ds \right] |x|^p \\
	& \leq \frac{C}{(T-t)^{p-1}}  (1 + |y|^n ) |x|^p,
\end{split}
\end{equation}	
due to estimate (\ref{moment-estimate}) and Assumption~\ref{A2}. As a result, we also expect $v(t,\cdot)$ to satisfy a polynomial growth condition. More precisely, we have the following result; its proof is given in Section~\ref{a-priori}. 

\begin{proposition} \label{thm-uniqueness}  \leavevmode Under Assumptions~\ref{A1} and~\ref{A2} the singular terminal value problem
\begin{equation} \label{pde-v}
	\left\{\begin{aligned}	&{-\partial_t v}(t,y)-\mathcal L v(t,y)- F(y,v(t,y))=0,    & (t,y)\in[0,T)\times\mathbb R^d,&\\
&\lim_{t\rightarrow T}v(t,y)=+\infty  & \text{locally uniformly on $\mathbb R^d$},&
\end{aligned}\right.
\end{equation}
with the nonlinearity $F$ given in~\eqref{nonlinearity} admits at most one nonnegative viscosity solution in \[
	C_{poly}([0,T^-]\times\mathbb R^d).
\] 
If such a viscosity solution exists, then it satisfies the following a priori estimates for $(t,y)\in[0,T)\times\mathbb R^d$:
\begin{equation} \label{a-priori-bounds}
	E\left[\frac{e^{-\theta(T-t)}}{\sqrt[\beta]{\int_t^T\frac{1}{\eta(Y_s^{t,y})^\beta}\,ds}}\right]\leq v(t,y)\leq \frac{1}{(T-t)^p}E\left[\int_t^T\eta(Y_s^{t,y})+(T-s)^p\lambda(Y_s^{t,y})\,ds\right].
\end{equation}
\end{proposition}

Our main contribution is the proof of the existence of a classical solution to the singular terminal value problem \eqref{pde-v}. This is achieved under assumptions (A1)--(A4).
The boundedness away from zero of $\eta$, along with the existence and boundedness of $\mathcal L\eta$ will for instance be used to derive the precise asymptotic behavior of the solution near the terminal time. The main result of this paper is:

\begin{theorem} \label{thm-existence-GHS}  \leavevmode 
Under Assumptions~\ref{A1} and~\ref{A2} and the conditions (A1)--(A4), the singular terminal value problem \eqref{pde-v} admits a nonnegative classical solution
\[
	v \in C^\alpha([0,T^-];D(\mathcal L))\cap C^{1+\alpha}([0,T^-];C(\mathbb R^d)).
\]
\end{theorem}

The proof of this theorem will be carried out in Section 4 below.\medskip

Our last result is a verification theorem; its proof is given in Section~\ref{sec-verification}. The singularity at the terminal time prevents a straightforward application of the standard verification arguments. Instead, we first use Krylov's generalized It\^o formula~\cite[Theorem~2.10.1, p.~122]{Krylov80} to establish optimality away from the terminal time and then use the a priori estimate~\eqref{a-priori-bounds} to prove optimality on the whole time interval.  Here we require that the strong solution is locally $L^q$-integrable along with its weak derivatives for some $q>d+2$ to guarantee that the parabolic Sobolev embedding theorem~\cite[Lemma~II.3.3]{LSU68} applies.

\begin{proposition} \label{thm-verifcation}
Under Assumptions~\ref{A1} and~\ref{A2} and the conditions (A1) and (A3), for some $q>d+2$, let
\[
v\in W^{1,2}_{q,loc}((0,T^-)\times\mathbb R^d)\cap C_{poly}([0,T^-]\times\mathbb R^d)
\]
be a nonnegative strong solution to \eqref{pde-v}.  Then, the value function of the control problem \eqref{value-function} is given by $V(t,y,x)=v(t,y)|x|^p$, and the optimal control $(\xi^*,\pi^*)$ is given in feedback form by
\begin{equation} \label{optimal-control}
		\xi_s^*= \frac{v(s,Y_s^{t,y})^\beta}{\eta(Y_s^{t,y})^\beta}X_s^* \quad \text{ and } \quad \pi_s^*=\frac{v(s,Y_s^{t,y})^\beta}{\gamma(Y_s^{t,y})^\beta+v(s,Y_s^{t,y})^\beta}X_{s-}^*.
\end{equation}
In particular, the resulting optimal portfolio process $(X^*_s)_{s\in[t,T]}$ is given by
\begin{equation} \label{optimal-position-process}	X_s^*=x\exp\left(-\int_t^s\frac{v(r,Y_r^{t,y})^\beta}{\eta(Y_r^{t,y})^\beta}\,dr\right)\prod_{t<r\leq s}^{\Delta N_r\neq 0}\left(1-\frac{v(t,Y_r^{t,y})^\beta}{\gamma(Y_r^{t,y})^\beta+v(t,Y_r^{t,y})^\beta}\right).
\end{equation}
\end{proposition}

Since $C^\alpha([0,T^-];D(\mathcal L))\cap C^{1+\alpha}([0,T^-];C(\mathbb R^d)) \subset W^{1,2}_{q,loc}((0,T^-)\times\mathbb R^d)\cap C_{poly}([0,T^-]\times\mathbb R^d)$  we immediately obtain the following corollary to Theorem \ref{thm-existence-GHS} and Proposition \ref{thm-verifcation}. 

\begin{corollary}
	Assume that the assumptions of Theorem \ref{thm-existence-GHS} hold and let $v$ be the unique classical solution to the singular terminal value problem \eqref{pde-v}. Then the optimal liquidation strategy is given in feedback form by (\ref{optimal-control}).  
\end{corollary}  


\section{Comparison principle and a priori estimates} \label{a-priori}

In this section we prove Proposition \ref{thm-uniqueness}. The proof is based on the following comparison principle that is itself a consequence of the comparison principle given in the Appendix for viscosity sub- and supersolutions to parabolic equations with \textit{finite} terminal values and monotone nonlinearities.  

\begin{lemma} \label{lemma-comparison principle}
	Let $\underline v,\overline v\in C_{poly}([0,T^-]\times\mathbb R^d)$ be a nonnegative viscosity sub- and a nonnegative viscosity supersolution to~\eqref{pde-v}, respectively, such that
\begin{equation*}
	\lim_{t\rightarrow T}\overline v(t,y)=+\infty \quad \text{locally uniformly on $\mathbb R^d$.}
\end{equation*}
Then, 
\[
	\underline v\leq\overline v \qquad \text{in} \quad [0,T)\times\mathbb R^d.
\] 
In particular, there exists at most one nonnegative viscosity solution in $C_{poly}([0,T^-]\times\mathbb R^d)$ to problem~\eqref{pde-v}.
\end{lemma}

\begin{proof}
	Due to the time-homogeneity of the PDE in~\eqref{pde-v}, viscosity (super-/sub-)solutions stay viscosity (super-/sub-)solutions when shifted in time. The idea is therefore to separate the singularities to have finite values to compare. 
	
More precisely, we define, for any $\delta>0$, the difference function $w:[0,T-\delta)\times\mathbb R^d\rightarrow \mathbb R$ by
\[
	w(t,y)= \overline v(t+\delta,y) -\underline v(t,y).
\]
A direct computation shows that $v\mapsto F(\cdot,v)$ is decreasing on $[0,\infty)$. In fact, both $\partial_v F$ and $\partial_v^2 F$ are nonpositive on $\mathbb R^d\times[0,\infty)$. Hence, by Lemma~\ref{lemma-linearization}, $w$ is on $[0,T-\delta)\times\mathbb R^d$ a viscosity supersolution to 
\begin{equation*}
-w_t(t,y)-\mathcal Lw(t,y)-l(t,y)w(t,y)=0,
\end{equation*}
where 
\[
	l(t,y):=1_{\overline v(t+\delta,y)\neq\underline v(t,y)}\frac{F(y,\overline v(t+\delta,y))-F(y,\underline v(t,y))}{\overline v(t+\delta,y)-\underline	v(t,y)}.
\]
By the first order Taylor approximation of $F$ in $v$ at $\overline v(t+\delta,y)$, along with $\partial_v F\leq 0$ and $\partial_v^2F\leq 0$, we obtain that
\[
	-l(t,y)w(t,y)\leq-\partial_v F(y,\overline v(t+\delta,y))w(t,y).
\]
In terms of the continuous coefficient
\[
	\tilde l(t,y):=\partial_v F(y,\overline v(t+\delta,y))\leq 0,
\]
it follows that $w$ is on $[0,T-\delta)\times\mathbb R^d$ also a viscosity supersolution to the equation
\begin{equation} \label{diff-pde}
	-\partial_t w(t,y)-\mathcal Lw(t,y)-\tilde l(t,y) w(t,y)=0,
\end{equation}
for which the assumptions of Theorem~\ref{theorem-finite-comparison} hold with $\mu=0$. Hence, for all $0<t\leq s<T-\delta$ and $y\in\mathbb R^d$ Theorem~\ref{theorem-finite-comparison} yields
\begin{equation} \label{std-comparison}
w(t,y)\geq E\left[w(s,Y_s^{t,y})\exp\left(\int_t^s \tilde l(r,Y_r^{t,y})\,dr\right)\right],
\end{equation}
where the RHS is the Feynman-Kac viscosity solution \cite[Theorem~3.2]{Pardoux99} to~\eqref{diff-pde} on $[0,s]\times\mathbb R^d$ with terminal value $w(s,\cdot)$. 

Since $\overline{v}\geq 0$, $\tilde l \leq 0$, and $E[\sup_{s\in[t,T-\delta]} \underline v(s,Y_s^{t,y})]<\infty$ due to $\underline v\in C_{poly}([0,T-\delta]\times\mathbb R^d)$, we can apply Fatou's lemma to the expectation in~\eqref{std-comparison} as $s\rightarrow T-\delta$ to obtain
\begin{equation} \label{123}
w(t,y)\geq E\left[\liminf_{s\rightarrow T-\delta}w(s,Y_s^{t,y})\exp\left(\int_t^s \tilde l(r,Y_r^{t,y})\,dr\right)\right].
\end{equation}
The sample paths of $(Y_s^{t,y})_{s\in[t,T]}$ are bounded a.\,s., and hence,
\[
	\lim_{s\rightarrow T-\delta} w(s,Y_s^{t,y})=\lim_{s\rightarrow T-\delta} \overline v(s+\delta,Y_s^{t,y})-\underline v(s,Y_s^{t,y})=+\infty \quad \mbox{a.s.} 
\]	
because $\lim_{s\rightarrow T-\delta} \overline v(s+\delta, \cdot) = + \infty$ uniformly on compact sets and $\underline v \in C_{poly}([0,T-\delta]\times\mathbb R^d)$. The limit inferior in~\eqref{123} is therefore a.s.\ nonnegative. Hence,
\[\overline v(t+\delta,y)-\underline v(t,y)\geq 0.\] Finally, by letting $\delta\rightarrow 0$ we conclude $\overline v-\underline v\geq0$ on $[0,T)\times\mathbb R^d$ by continuity of $\overline v$.
\end{proof}

	 The comparison principle establishes the uniqueness statement in Proposition \ref{thm-uniqueness}. It also allows us to establish the a priori estimates~\eqref{a-priori-bounds}.
	 
\begin{proof}[Proof of Proposition~\ref{thm-uniqueness}]
	We show that the lower (upper) estimate in \eqref{a-priori-bounds}---denoted in the following by $\underline v$ (respectively, $\overline v$)---is a viscosity subsolution (supersolution) to~\eqref{pde-v}. The assertion then follows from the comparison principle established in the preceding lemma. 

	First, note that $\underline v\in C_{poly}([0,T^-]\times\mathbb R^d)$. In fact, by Jensen's inequality 
\begin{align*}
	0\leq \underline v(t,y)&=\frac{e^{-\theta(T- t)}}{(T-t)^{1/\beta}}E\left[\left(\frac{1}{T-t}\int_t^T\frac{1}{\eta(Y_s^{t,y})^{\beta}}\,ds\right)^{-1/\beta}\right]\\
	&\leq \frac{e^{-\theta(T-t)}}{(T-t)^{(\beta+1)/\beta}}E\left[\int_t^T\eta(Y_s^{t,y})\,ds\right],
\end{align*}
and hence the polynomial growth of $\underline v$ in $y$ uniformly away from the terminal time follows from~\eqref{moment-estimate} and the polynomial growth of $\eta$. The continuity along with the polynomial growth of $1/\eta$ guarantees continuity of  $\underline v$, due to Vitali's convergence theorem as pointed out in  (\ref{Vitali}).
%

	To establish the subsolution property of $\underline v$, let $\varphi\geq \underline v$ be a smooth test function on $[0,T)\times\mathbb R^d$ such that $\varphi(t,y)=\underline v(t,y)$ for some $(t,y)\in[0,T)\times\mathbb R^d$. Moreover, we define the stopping time $\tau=\inf\{s\in[t,T]:|Y_s^{t,y}-y|\geq 1\}$. By uniqueness $Y_r^{t,y}=Y_r^{s,Y_s^{t,y}}$ for all $t\leq s\leq r\leq T$. Hence, by the definition of $\underline v$ and the Markov property of $Y^{t,y}$,
\[
	\underline v(s\wedge\tau,Y_{s\wedge\tau}^{t,y})=E\left[\left.\frac{ e^{-\theta(T- (s\wedge\tau))}}{\sqrt[\beta]{\int_{s\wedge\tau}^T\eta(Y_r^{t,y})^{-\beta}\,dr}}\right|\mathcal F_{s\wedge\tau}\right].
\]
Because $\varphi\geq \underline v$, $\varphi(t,y)=\underline v(t,y)$, and the tower rule, it holds for $t<s<T$,
\begin{align*}
	0&\leq E[\varphi(s\wedge\tau,Y_{s\wedge\tau}^{ t, y})-\underline v(s\wedge\tau,Y_{s\wedge\tau}^{t,y})]\\
	&= E[\varphi(s\wedge\tau,Y_{s\wedge\tau}^{t,y})-\varphi(t,y)] -
	E\left[\frac{e^{-\theta(T-( s\wedge\tau))}}{\sqrt[\beta]{\int_{s\wedge\tau}^T\eta(Y_r^{t,y})^{-\beta}\,dr}}\right]+E\left[\frac{e^{-\theta(T- t)}}{\sqrt[\beta]{\int_t^T\eta(Y_r^{t,y})^{-\beta}\,dr}}\right].
\end{align*}
Dividing by $s-t$, using It\^o's formula and the integral representation of increments of the function  
$r \mapsto \frac{e^{-\theta(T-r)}}{\sqrt[\beta]{\int_{r}^T\eta(Y_u^{t,y})^{-\beta}\,du}}$, 
and noticing that the stochastic integral being stopped at $\tau$ is a true martingale, we obtain
\begin{multline*} 
0\leq  E\left[\frac{1}{s-t}\int_t^{s\wedge\tau}\partial_t\varphi(r,Y_r^{t,y})+\mathcal L\varphi(r,Y_r^{t,y})\,dr\right]\\
-E\Bigg[\frac{1}{s-t}\int_t^{s\wedge\tau} \frac{1}{\beta\eta(Y_r^{t,y})^\beta} \frac{e^{-\theta(T-r)}}{\left(\int_{r}^T\eta(Y_u^{t,y})^{-\beta}\,du\right)^{(\beta+1)/\beta}}\,dr \Bigg]\\
-E\Bigg[\frac{1}{s-t}\int_t^{s\wedge\tau}\frac{\theta e^{-\theta(T- t)}}{\sqrt[\beta]{\int_r^T\eta(Y_u^{t,y})^{-\beta}\,du}}\,dr\Bigg].
\end{multline*}
By Jensen's inequality, 
for $r\in[t,s]$,
\[
	\sqrt[\beta]{\frac{1}{\int_r^T\eta(Y_u^{t,y})^{-\beta}\,du}} \leq 
	\sqrt[\beta]{\frac{1}{(T-r)^2} \int_r^T \eta(Y_u^{t,y})^{\beta}\,du} \leq \frac{1}{\sqrt[\beta]{T-s}}\sup_{t \leq u \leq T} \eta(Y_u^{t,y}).
\]
Since $\eta$ is of polynomial growth it follows from the equation (\ref{moment-estimate}) that we can apply the dominated convergence theorem to the third term above when letting $s \to t$. Since $|Y^{t,x}_r - y| \leq 1$ for $r \in [t,\tau]$ and because $\eta^{-1}$ is bounded on compact domains similar arguments show that we can apply the dominated convergence theorem also to the second term. As $\tau > t$, the fundamental theorem of calculus yields    
\[
	0\leq \partial_t\varphi(t,y)+\mathcal L \varphi(t,y) - \frac{1}{\beta \eta(y)^\beta}E\left[\left(\frac{e^{-\theta(T-t)/(\beta+1)}}{\sqrt[\beta]{\int_t^T\eta(Y_u^{t,y})^{-\beta}\,du}}\right)^{\beta+1}\right]- \theta\underline v(t,y).
\]
Using Jensen's inequality and $\beta+1>1$ we obtain,
\[
0\leq\partial_t\varphi(t,y)+\mathcal L \varphi(t,y) - \frac{\underline v(t,y)^{\beta+1}}{\beta \eta(y)^\beta}- \theta\underline v(t,y).
\]
Hence, from the definition~\eqref{nonlinearity} of $F$ and $\lambda,\gamma,\underline v\geq0$ it is seen,
\[
	-\partial_t\varphi(t,y)-\mathcal L \varphi(t,y)-F(y,\underline v(t,y))\leq 0.
\]

	Next, we verify the supersolution property of $\overline v$. Because $\eta$ and $\lambda$ are of polynomial growth it follows $\overline v\in C_{poly}([0,T^-]\times\mathbb R^d)$ from \eqref{Vitali}. Again an application of It\^o's formula and Leibniz's rule similar as above yields that for every smooth test function $\varphi\leq \overline v$ on $[0,T)\times\mathbb R^d$ such that $\varphi(t,y)=\overline v(t,y)$ for some $(t,y)\in[0,T)\times\mathbb R^d$,
\[
	0\geq \partial_t\varphi(t,y)+\mathcal L\varphi(t,y)+\frac{\eta(y)}{(T-t)^p}+\lambda(y)-\frac{p\overline v(t,y)}{T-t}.
\]
From the definition~\eqref{nonlinearity} of $F$, since $\overline v\geq 0$,
\[
	-F(y,\overline v(t,y))\geq -\lambda(y)+\frac{|\overline v(t,y)|^{\beta+1}}{\beta \eta(y)^\beta}.
\]
Hence, since $p=(\beta+1)/\beta$ and by setting $u(t,y)=(T-t)^{1/\beta}\overline v(t,y)/\eta(y)$,
\begin{equation} \label{tedious}
\begin{aligned}
	-\partial_t\varphi(t,y)-\mathcal L\varphi(t,y)-F(y,\overline v(t,y)) &\geq \frac{\eta(y)}{(T-t)^p}-\frac{p\overline v(t,y)}{T-t}+\frac{|\overline v(t,y)|^{\beta+1}}{\beta \eta(y)^\beta}\\
	&= \frac{\eta(y)}{(T-t)^{\frac{\beta+1}{\beta}}} \left(1-\tfrac{(\beta+1)}{\beta}u(t,y)-\tfrac{1}{\beta} |u(t,y)|^{\beta+1}\right).
\end{aligned}
\end{equation}
Using the fact that the map $u\mapsto 1-\frac{\beta+1}{\beta}u+\frac{1}{\beta}|u|^{\beta+1}$ is nonnegative on $\mathbb R$ as it attains its minimum at $1$, we conclude
\[
-\partial_t\varphi(t,y)-\mathcal L\varphi(t,y)-F(t,\overline v(t,y))\geq 0. \qedhere
\]
\end{proof}

	We close this section with a further application of the comparison principle. Under the conditions (A3) and (A4) it establishes the precise asymptotic behavior of a viscosity solution at the terminal time. This observation will be the starting point of the next section.
	
\begin{corollary} \label{cor-a-priori-estimate} 
Let $v\in C_{poly}([0,T^-]\times\mathbb R^d)$ be a nonnegative viscosity solution to problem~\eqref{pde-v}. If the assumptions (A3) and (A4) hold, then $v$ satisfies the following asymptotic behavior: 
\begin{equation} \label{order-of-error}
	(T-t)^{1/\beta}v(t,y)=\eta(y)+ O(T-t) \quad \text{uniformly in $y$ as $t\rightarrow T$.}
\end{equation}
\end{corollary}

\begin{proof}
The statement is proved by identifying a sub- and a supersolution with the desired asymptotics. Due to (A4), the quantity $\|\mathcal L\eta\|$ is well-defined and finite, hence $\delta :=\kappa_0/\|\mathcal L\eta\| > 0$. We verify below that
\begin{equation*}
	 \check v(t,y):=\frac{\eta(y)-\|\mathcal L\eta\|(T-t)}{e^{\theta (T-t)}(T-t)^{1/\beta}} \quad \mbox{and} \quad 
	 \hat v(t,y):=\frac{\eta(y)+\frac{1}{2}\|\mathcal L\eta\|(T-t)}{(T-t)^{1/\beta}}+(T-t)\|\lambda\|
\end{equation*}
are a nonnegative classical sub- and supersolution to~\eqref{pde-v} on $[T-\delta,T)\times\mathbb R^d$, respectively, where nonnegativity follows from $\eta \geq \kappa_0 > 0$ by (A3). Hence, \eqref{order-of-error} follows from the comparison principle. 

Specifically, let us fix $(t,y)\in[T-\delta,T)\times\mathbb R^d$. To verify the supersolution property of $\hat v$, we first obtain by a direct computation,
\begin{equation} \label{derivative-v-hat}
	-\partial_t\hat v(t,y)-\mathcal L \hat v(t,y)=-\frac{\eta(y)+\frac{1-\beta}{2}\|\mathcal L \eta\|(T-t)+\beta\mathcal L \eta(T-t)}{\beta(T-t)^{(\beta+1)/\beta}}+\|\lambda\|.
\end{equation}
Recalling the definition~\eqref{nonlinearity} of $F$, we have since $\hat v\geq0$,
\[
	-F(y,\hat v(t,y))\geq -\lambda(y)+\frac{|\hat v(t,y)|^{\beta+1}}{\beta \eta(y)^\beta}.
\]
Next, we apply Bernoulli's inequality in the form $(u+v+w)^{\beta+1}\geq u^{\beta+1}(1+v/u)^{\beta+1}\geq u^{\beta+1}+(\beta+1)u^\beta v$ for $u,v,w\geq 0$ to the term $|\hat v(t,y)|^{\beta+1}$ and obtain
\begin{equation} \label{nonlinearity-v-hat}
	-F(y,\hat v(t,y))\geq-\lambda (y)+\frac{\eta(y)^{\beta+1}+(\beta+1)\eta(y)^\beta\frac{1}{2}\|\mathcal L\eta\|(T-t)}{\beta\eta(y)^\beta (T-t)^{(\beta+1)/\beta}}.
\end{equation}
Hence, adding \eqref{derivative-v-hat} and \eqref{nonlinearity-v-hat} yields,
\[
	-\partial_t\hat v(t,y)-\mathcal L \hat v(t,y)-F(y,\hat v(t,y))\geq \|\lambda\|-\lambda(y)+ \frac{\|\mathcal L\eta\|-\mathcal L \eta(y)}{(T-t)^{1/\beta}}\geq 0.
\]

Next, we verify the subsolution property of $\check v$. By a direct computation
\begin{equation} \label{derivative-v-check}
	-\partial_t\check v(t,y)-\mathcal L \check v(t,y)=-\frac{\eta(y)+(\beta-1)\|\mathcal L\eta\|(T-t)+\beta\mathcal L\eta(y)(T-t)}{\beta e^{\theta(T-t)}(T-t)^{(\beta+1)/\beta}}-\theta \check v_t(t,y).
\end{equation}
On the other hand, since $\lambda,\gamma\geq 0$, and $\check v\geq 0$ on $[T-\delta)\times\mathbb R^d$,
\[
	-F(y,\check v(t,y))\leq \frac{|\check v(t,y)|^{\beta+1}}{\beta\eta(y)^\beta}+\theta\check v(t,y).
\]
This time, recalling that $\delta$ is chosen such that $\eta(y)\geq \|\mathcal L\eta\|(T-t)$, we estimate the term $|\check v(t,y)|^{\beta+1}$ by the fact $(u-v)^{\beta+1}=u^{\beta+1}(1-v/u)^{\beta+1}\leq u^{\beta+1}-u^\beta v$ for $u\geq v\geq 0$ and obtain
\begin{equation} \label{nonlinearity-v-check}
-F(y,\check v(t,y))\leq\frac{\eta(y)-\|\mathcal L\eta\|(T-t)}{\beta e^{\theta(\beta+1)(T-t)}(T-t)^{(\beta+1)/\beta}}+\theta \check v(t,y).
\end{equation}
Finally, adding \eqref{derivative-v-check} and \eqref{nonlinearity-v-check}, and using $\beta>0$ yields,
\begin{equation*}
-\partial_t\check v(t,y)-\mathcal L \check v(t,y)-F(t,\check v(t,y))\leq 0. \qedhere
\end{equation*}
\end{proof}


\section{Existence of a classical solution} \label{sec-local-existence}

	In this section we prove Theorem \ref{thm-existence-GHS} and hence assume throughout that (A1)--(A4) hold. Our existence proof is based on the explicit asymptotic behavior established in Corollary~\ref{cor-a-priori-estimate}. It tells us the solution must be of the form 
\begin{equation} \label{straightforward-ansatz}
	\qquad\qquad\qquad  v(T-t,y)= \frac{\eta(y)+\tilde u(t,y)}{t^{1/\beta}}, \quad \tilde u(t,y)=  O(t) \text{ uniformly in $y$ as $t\rightarrow0$},
\end{equation}
where we reversed the time variable as we will do for the rest of this subsection. For reasons that will become clear later, it will be more convenient to choose the following equivalent ansatz: 
\begin{equation} \label{educated-ansatz}
	\qquad\qquad\qquad v(T-t,y)= \frac{\eta(y)}{t^{1/\beta}}+\frac{u(t,y)}{t^{1+1/\beta}}, \quad u(t,y)=  O(t^2) \text{ uniformly in $y$ as $t\rightarrow0$}.
\end{equation}
Plugging the asymptotic ansatz into~\eqref{pde-v} results in a semilinear parabolic equation for $u$ with finite initial condition, but with a singularity in the nonlinearity. This motivates the following lemma.

\begin{lemma} \label{lemma-separation}
	If for some $\delta>0$ a function $u\in  C^\alpha({[0,\delta]};D(\mathcal L))\cap C^{1+\alpha}({[0,\delta]};C(\mathbb R^d))$ satisfies
\begin{equation}  \label{growth-of-u}
	|u(t,y)|\leq t\eta(y), \quad t\in[0,\delta],\, y\in\mathbb R^d,
\end{equation}
and solves the equation
\begin{equation} \label{pde}
\begin{split} 
	\partial_tu(t,y)&=\mathcal L u(t,y)+t\mathcal L \eta(y)+t^{p}\lambda(y)-\frac{\eta(y)}{\beta}\sum_{k=2}^\infty\dbinom{\beta +1}{k} \left(\frac{u(t,y)}{t\eta(y)}\right)^k \\&\quad+\frac{\theta t^p\gamma(y)(t\eta(y)+u(t,y))}{\sqrt[\beta]{(t^p\gamma(y))^{\beta}+|t\eta(y)+u(t,y)|^{\beta}}}-\theta(t\eta(y)+u(t,y)), \quad t>0\,,y\in\mathbb R^d,
\end{split}
\end{equation}
then a local solution $v\in C^\alpha([T-\delta,T^-];D(\mathcal L))\cap C^{1+\alpha}([T-\delta,T^-];C(\mathbb R^d))$ to problem~\eqref{pde-v} is given by
\begin{equation*}
	v(t,y)= \frac{\eta(y)}{(T-t)^{1/\beta}}+\frac{u(T-t,y)}{(T-t)^{1+1/\beta}}.
\end{equation*}
\end{lemma}

\begin{proof}
The statement is verified by plugging the ansatz into~\eqref{pde-v}, multiplying by $t^p=t^{(\beta+1)/\beta}$, and by using the binomial series for the term
\begin{align*}
 t^p\frac{|v(T-t,y)|^{\beta+1}}{\beta\eta(y)^\beta} =\frac{\eta(y)}{\beta }\left|1+\frac{u(t,y)}{t\eta(y)}\right|^{\beta+1} =\frac{\eta(y)}{\beta }\sum_{k=0}^\infty\dbinom{\beta +1}{k} \left(\frac{u(t,y)}{t\eta(y)}\right)^k
\end{align*}
to see that the first two terms of the series cancel out. The growth condition \eqref{growth-of-u} guarantees that the binomial series does indeed converge.
\end{proof}

\begin{remark}
	The reason for choosing the ansatz~\eqref{educated-ansatz} is that the series in \eqref{pde} starts at $k=2$, and not at $k=1$. This will be crucial for the fixed point argument below. The more straightforward ansatz \eqref{straightforward-ansatz} results in the equation
\begin{equation*}
	\partial_t\tilde u=\mathcal L \tilde u+\mathcal L\eta+t^{\frac{1}{\beta}}\lambda-\frac{\eta}{\beta t}\left(\frac{\tilde u}{\eta}+\sum_{k=2}^\infty\binom{\beta+1}{k}\left(\frac{\tilde u}{\eta}\right)^k\right)+\frac{\theta t^{\frac{1}{\beta}}\gamma(\eta+\tilde u)}{\sqrt[\beta]{t\gamma^\beta+|\eta+\tilde u|^\beta}}-\theta(\eta+\tilde u),
\end{equation*}
for which we have no analogues to Lemma~\ref{lemma-f-continuous} and Lemma~\ref{lemma-locally-lip} below, due to the term $t^{-1}\tilde u$.
\end{remark}

	 We will solve equation \eqref{pde} using the semigroup approach for parabolic equations in Banach spaces; we refer to the monograph by Lunardi \cite{Lunardi95} as the standard reference. To this end, we interpret \eqref{pde} as an evolution equation
\begin{equation} \label{abstract-pde}
	u^\prime(t)=\mathcal L u(t)+f(t,u(t)), \quad t>0;\quad u(0)=0,
\end{equation}
in the Banach algebra $U:=C(\mathbb R^d)$ of bounded continuous functions endowed with the supremum norm $\|\cdot\|$, where the nonlinearity $f$ is given by
\begin{equation*} 
	f(t,u)=t\mathcal L\eta+t^{p}\lambda-\frac{\eta}{\beta}\sum_{k=2}^\infty\dbinom{\beta +1}{k} \left(\frac{u}{t\eta}\right)^k + \frac{\theta t^p\gamma(t\eta+u)}{\sqrt[\beta]{(t^p\gamma)^{\beta}+|t\eta+u|^{\beta}}}-\theta(t\eta+u).
\end{equation*}

	 The general theory suggests to look first for a local \textit{mild solution} of \eqref{abstract-pde}. That is, to show there is a fixed point $u$ of the integral operator $\Gamma$ defined in $C([0,\delta];U)$ by
\begin{equation} \label{operator}
	\Gamma(u)(t)=\int_0^te^{(t-s)\mathcal L}f(s,u(s))\,ds, \quad 0\leq t\leq\delta,
\end{equation}
if $\delta>0$ is small enough where $\{e^{t\mathcal L} : t \geq 0\}$ is the analytic semigroup generated by $\mathcal L$ in~$U$. Regularity of the mild solution $u$ will then follow from analyticity of the semigroup and H\"{o}lder continuity of $t\mapsto f(t,u(t))$. 

The singular behavior of $f$ near $t=0$ prevents us from directly applying general theory. In fact, the operator $\Gamma$ is not defined on the whole space $C([0,\delta];U)$, and its domain is not closed with respect to the supremum norm. We overcome these difficulties by carrying out the usual contraction argument with respect to an appropriate weighted norm on $C([0,\delta];U)$. 

	In order to guarantee that the function $t\mapsto f(t,u(t))$ behaves well at $t=0$ it seems reasonable to restrict the set of potential mild solutions to those functions $u\in C([0,\delta];U)$ such that $u(t)= o(t)$ as $t\rightarrow 0$. Yet, there is no nice norm making this set of functions a Banach space. Recalling \eqref{order-of-error} however, we actually expect the slightly stronger condition $u(t)=O(t^2)$ as $t\rightarrow 0$ to be satisfied. This suggests to view $\Gamma$ as an operator acting in the space
\begin{equation*}
	E=\big\{u\in C([0,\delta];U): u(t)= O(t^2) \text{ as } t\rightarrow 0 \big\},
\end{equation*}
endowed with the weighted norm
\begin{equation*}
	\left\|u\right\|_E=\sup_{0<t\leq \delta}\left\|t^{-2}u(t)\right\|.
\end{equation*}  

\begin{lemma}
	The vector space $E$ endowed with the norm $\left\|\,\cdot\,\right\|_E$ is a Banach space.
\end{lemma}

	 The next lemma shows in particular that the integral operator $\Gamma$ given in \eqref{operator} is well-defined on the closed ball 
\[
	 \overline B_E(\kappa_0/\delta):= \big\{u\in E:\left\|u\right\|_E\leq \kappa_0/\delta \big\}.
\]
\begin{lemma} \label{lemma-f-continuous}
	Let $R>0$ and $\delta\in{(0,\kappa_0/R]}$. 
	\begin{itemize}
		\item[(i)] For every $u\in \overline B_E(R)$, the function $f(\,\cdot\,,u(\,\cdot\,))$ belongs to  $C([0,\delta];U)$. In particular, the operator $\Gamma$ defined in \eqref{operator} is well defined on $\overline B_E(R)$. 
		\item[(ii)] If $u\in \overline B_E(R)\cap C^\alpha([0,\delta];U)$ for some $\alpha \in (0,1)$, then $f(\,\cdot\,,u(\,\cdot\,))$ is $\alpha$-H\"older continuous, i.e., belongs to $C^\alpha([0,\delta];U)$.
	\end{itemize}
\end{lemma}

\begin{proof}
	For $u\in\overline B_E(R)$ we consider the functions $g:{[0,\delta]}\rightarrow U$ and $h:{[0,\delta]}\times U\rightarrow U$ given by
\begin{equation*}
	g(t)=\sum_{k=2}^\infty \dbinom{\beta +1}{k}\left(\frac{u(t)}{t\eta}\right)^k \quad \text{and} \quad 
h(t,w)=\frac{t^p\gamma w}{\sqrt[\beta]{(t^p\gamma)^{\beta}+|w|^{\beta}}}-w,
\end{equation*}
so that we may decompose $f(t,u(t))$ in the following way:
\begin{equation} \label{decomposition}
	f(t,u(t))=t\mathcal L\eta+t^{p}\lambda-(p-1)\eta g(t) +\theta h(t,t\eta+u(t)).
\end{equation}
The assumption $\delta\leq\kappa_0 /R$ guarantees that the series defining $g(t)$ converges in $U$ since then
\begin{equation*}
	\left\|\frac{u(t)}{t\eta}\right\| \leq\frac{t^2 R}{t\kappa_0}\leq\frac{\delta R}{\kappa_0}\leq 1, \quad t\in[0,\delta].
\end{equation*}
In view of \eqref{decomposition} it will be sufficient to show that $g$ and $h(\,\cdot\,,\,\cdot\,\eta+u(\,\cdot\,))$ are continuous, or even $\alpha$-H\"{o}lder continuous if $u\in C^\alpha([0,\delta];U)$. For the latter note that $h$ is continuously differentiable on $(0,\delta]\times U$. In fact,
\begin{equation*}
	\left\|\partial_th(t,w)\right\|_{L(U)}=\left\|\frac{pt^{p-1}\gamma w|w|^\beta}{\sqrt[\beta]{ (t^p\gamma)^\beta+|w|^\beta}^{\beta+1}}\right\|\leq \left\|\frac{pt^{p-1}\gamma w|w|^\beta}{(t^p\gamma)^{\beta + 1}+pt^p\gamma|w|^\beta}\right\|\leq\frac{\|w\|}{t },
\end{equation*}
where we used Bernoulli's inequality and $\beta+1=p\beta$, and 
\begin{equation*}
	\left\|\partial_wh(t,w)\right\|_{L(U)}=\left\|\frac{(t^p\gamma)^{\beta+1}}{\sqrt[\beta]{ (t^p\gamma)^\beta+|w|^\beta}^{\beta+1}}-1\right\|\leq 1. 
\end{equation*}
Hence, for all $0\leq t\leq s\leq\delta$,
\begin{align*}
	\|h(t,t\eta+u(t))-h(s,s\eta+u(s))\|&\leq \frac{\|t\eta+u(t)\|}{t}|t-s|+ { \| \eta\| |t-s|}+\|u(t)-u(s)\| \\
	&\leq (2 \|\eta\|+\kappa_0)|t-s|+\|u(t)-u(s)\|.
\end{align*}

	In order to establish the continuity of $g$, notice that for every $k\geq 2$ and $0\leq t\leq s\leq\delta$ it holds that
\begin{equation} \label{first-estimate}
\begin{aligned}
	\left\|\left(\frac{u(t)}{t\eta}\right)^k\right. -&\left.\left(\frac{u(s)}{s\eta}\right)^k\right\|\leq \left\|\left(\frac{u(t)}{t\eta}\right)^k-\left(\frac{u(t)}{s\eta}\right)^k\right\| + \left\|\left(\frac{u(t)}{s\eta}\right)^k-\left(\frac{u(s)}{s\eta}\right)^k\right\| 
\\
&\leq \frac{\left\|u(t)\right\|^k}{\kappa_0^k} \left|\frac{1}{t^k}-\frac{1}{s^k}\right| +\frac{1}{s^k \kappa_0^k} \left\|u(t)-u(s)\right\|\sum_{l=0}^{k-1}\left\|u(t)\right\|^l\left\|u(s)\right\|^{k-1-l}
\\
&\leq \frac{t^{2k}R^k}{t^ks^k \kappa_0^k} \left|t^k-s^k\right| +\frac{R^{k-1}}{s^k \kappa_0^k} \left\|u(t)-u(s)\right\|\sum_{l=0}^{k-1}t^{2l}s^{2(k-1-l)}\\
&\leq\frac{k\delta^{k-1}R^k}{\kappa_0^k}\left|t-s\right| +\frac{k\delta^{k-2}R^{k-1}}{\kappa_0^k}\|u(t)-u(s)\| \\
&\leq \frac{kR}{\kappa_0}|t-s|+\frac{k R}{\kappa_0^2}\left\|u(t)-u(s)\right\|.
\end{aligned}
\end{equation}
Using the identity $ k\binom{\beta+1}{k}=(\beta+1)\binom{\beta}{k-1}$ it follows that
\begin{equation*}
\|g(t)-g(s)\|\leq\frac{(2^\beta-1)(\beta+1)R}{\kappa_0}|t-s|+\frac{(2^\beta-1)(\beta+1)R}{\kappa_0^2}\|u(t)-u(s)\|.
\end{equation*}
Hence, $g$ is uniformly continuous and even $\alpha$-H\"{o}lder continuous if $u\in C^\alpha([0,\delta];U)$.
\end{proof}

	 The usual assumption on the nonlinearity to carry out the fixed point argument would be that $f(t,u)$ is locally Lipschitz continuous in $u$ uniformly in $t$. The next lemma proves an appropriate analogue to this assumption for our singular nonlinearity $f$. 
	
\begin{lemma} \label{lemma-locally-lip}
	For every $R>0$ there exists a constant $L>0$ independent of $\delta\in{(0,\kappa_0/R]}$ such that
\begin{equation*}
	\left\|f(t,u(t))-f(t,v(t))\right\|\leq L\left\|u(t)-v(t)\right\|, \quad   u,v\in\overline B_E(R),\, \, t\in[0,\delta].
\end{equation*}
\end{lemma}

\begin{proof} 
	Let $u,v\in\overline B_E(R)$ and $t\in[0,T]$. The proof of Lemma~\ref{lemma-f-continuous} shows that the function $h$ in~\eqref{decomposition} is nonexpanding in the second argument, and estimates similar to those in \eqref{first-estimate} yield
\begin{equation*}
	\left\|\left(\frac{u(t)}{t\eta}\right)^k- \left(\frac{v(t)}{t\eta}\right)^k\right\|\leq\frac{kR}{\kappa_0^2} \left\|u(t)-v(t)\right\|
\end{equation*}
for every $k\geq2$. Hence, using once more that $k\binom{\beta+1}{k}=(\beta+1)\binom{\beta}{k-1}$ and $(\beta+1)/\beta=p$ we conclude that
\begin{equation*}
	\|f(t,u(t))-f(t,v(t))\| \leq(p(2^\beta-1)\kappa_0^{-2}R\|\eta\|+\theta)\left\|u(t)-v(t)\right\|. \qedhere
\end{equation*}
\end{proof}

	 We are now ready to carry out the fixed point argument and to prove the desired regularity of the fixed point. In view of Lemma~\ref{lemma-separation} this then gives us a local solution to the problem~\eqref{pde-v}.
	
\begin{proposition} \label{prop-local-existence}
	Under assumptions (A1) and (A2), there exists a short-time solution $$u\in C^\alpha({[0,\delta]};D(\mathcal L))\cap C^{1+\alpha}({[0,\delta]};C(\mathbb R^d))$$ to the equation \eqref{pde} that satisfies the growth condition \eqref{growth-of-u}.
\end{proposition}

\begin{proof}

	We prove below that there exists $R>0$ and $\delta\in{(0,\kappa_0/R]}$ such that the operator~$\Gamma$ defined by \eqref{operator} has a fixed point $\overline u$ in $\overline{B}_E(R)$. 
	
	In order to see that this (local) mild solution to \eqref{abstract-pde} belongs to $C^\alpha({[0,\delta]};D(\mathcal L))\cap C^{1+\alpha}({[0,\delta]};U)$ notice first that $f(\,\cdot\,,\overline u(\,\cdot\,)) \in  C({[0,\delta]};U)$, due to Lemma~\ref{lemma-f-continuous}(i). Thus, it follows from  \cite[Proposition 4.2.1]{Lunardi95} that $\Gamma$ maps into $C^\alpha([0,\delta];U)$, for every $\alpha\in{(0,1)}$. Since $\overline u$ is a fixed point of $\Gamma$ it follows from Lemma~\ref{lemma-f-continuous}(ii) that $f(\,\cdot\,,\overline u(\,\cdot\,)) \in C^\alpha([0,\delta];U)$. 
	 
	Moreover, $\overline u\in E$  implies that $\overline u(0)= \mathcal L \overline u(0)+f(0,\overline u(0))\equiv 0$ belongs to the domain of $\mathcal{L}$. Along with the  H\"{o}lder continuity of $f(\,\cdot\,,\overline u(\,\cdot\,))$  it now follows from  \cite[Theorem 7.1.10(iv)]{Lunardi95} that 
$\overline u\in  C^\alpha({[0,\delta]};D(\mathcal L))\cap C^{1+\alpha}({[0,\delta]};U)$ as desired.

It remains to prove the existence of a fixed point of the operator $\Gamma$. In terms of
	$M=\sup_{0\leq t\leq 1}\|e^{t\mathcal L}\|_{L(U)}$
we claim that one can choose
\begin{equation*}
	R=2M \left(\|\mathcal L\eta\|+\|\lambda\|+\theta\|\eta\|\right) \quad\text{and}\quad\delta=\min\{\kappa_0/R,(2ML)^{-1},1\},
\end{equation*}
where $L > 0$ is the Lipschitz constant given by Lemma~\ref{lemma-locally-lip}. Since $\delta\leq \kappa_0/R$ the operator $\Gamma$ is  well-defined on $\overline B_E(R)$, due to Lemma~\ref{lemma-f-continuous}. To show that $\Gamma$ is a contraction with respect to $\|\cdot\|_E$, let $u,v\in\overline B_E(R)$. By the choice of $M$ it holds for every $t\in[0,\delta]$ that
\begin{align*} 
	\|\Gamma(u)(t)-\Gamma(v)(t)\|&\leq t M \sup_{s\in[0,t]}\|f(s,u(s))-f(s,v(s))\|\\ &\leq \delta M L\sup_{s\in[0,t]} \|u(s)-v(s)\|\\ &\leq \delta ML t^2\|u-v\|_E.
\end{align*}
Hence, 
\begin{equation*}
	\|\Gamma(u)-\Gamma(v)\|_E\leq \frac{1}{2} \|u-v\|_E.
\end{equation*}

	To show that $\Gamma$ maps $\overline B_E(R)$ into itself, note that since $\delta\leq1$ and $p>1$ one has that $s^p\leq s$ for all $s\in[0,\delta]$, and so it holds for every $t\in[0,\delta]$ that
\begin{align*}
	\|\Gamma(u)(t)\|&\leq \|\Gamma(u)(t)-\Gamma(0)(t)\|+\|\Gamma(0)(t)\|\\&\leq t^2\frac{R}{2}+tM \sup_{s\in[0,t]}\bigg\|s \mathcal L\eta+s^{p}\lambda +\frac{\theta s^p\gamma s\eta}{\sqrt[\beta]{(s^p\gamma)^\beta+(s\eta)^\beta}}- \theta s\eta\bigg\|\\
	&\leq t^2\frac{R}{2}+tM\sup_{s\in[0,t]}\left\{s\|\mathcal L\eta\|+s^p\|\lambda\|+\theta s\|\eta\|\right\}\\
	&\leq t^2R.
\end{align*}
The operator $\Gamma$ does therefore map $\overline B_E(R)$ contractive into itself. Hence, it has a unique fixed point $\overline u$ in $\overline B_E(R)$. 
\end{proof}

\begin{remark} \label{remark-pi-semigroup}
	In order to apply the above fixed point argument establishing a mild solution to~\eqref{abstract-pde} it is not needed that $\mathcal L$ generates an analytic semigroup. We use the analyticity of $\mathcal L$ to obtain maximal regularity results. We thank an anonymous referee for pointing out that one may possibly drop the additional assumptions (A1) and (A2) on the diffusion coefficients and interpret~$\mathcal L$ as the generator  $\tilde{ \mathcal L}$ of the $\pi$-continuous Markov transition semigroup induced by~$Y$ on~$C(\mathbb R^d)$; see \cite[Appendix~B.5]{FabbriGozziSwiech17} for a comprehensive overview of the theory of $\pi$-continuous semigroups. It is shown in \cite{Priola01} (see also \cite{CerraiGozzi95}) that the obtained mild solution is a $\pi$-strong approximation of strict solutions $u_n$. However, for the verification argument a $\pi$-strong approximation satisfying $u_n(t,\cdot)\in C^2(\mathbb R^d)$ is needed to apply It\^o's formula on $u_n$ and then use the dominated convergence. This may possibly be achieved as in~\cite{CerraiGozzi95} by establishing that $\{u\in C^2(\mathbb R^d): \tilde{\mathcal L} u\in C(\mathbb R^d)\}$ is a $\pi$-core for $\tilde{\mathcal L}$. The later is proven in \cite{DaPratoRoeckner11} under the assumption that $b$ and $\sigma$ have a continuous and bounded second derivative.   
\end{remark}

We are now ready to prove Theorem~\ref{thm-existence-GHS}.

\begin{proof}[Proof of Theorem~\ref{thm-existence-GHS}]
	In view of Lemma~\ref{lemma-separation} and Proposition~\ref{prop-local-existence} there exists a (unique) local classical, and hence mild solution 
\[
	u\in C^\alpha({[T-\delta,T^-]};D(\mathcal L))\cap C^{1+\alpha}({[T-\delta,T^-]};C(\mathbb R^d))
\]	
to~\eqref{pde-v}. In order to see that the local solution extends to a global solution 
\[
	v \in C^\alpha({[0,T^-]};D(\mathcal L))\cap C^{1+\alpha}({[0,T^-]};C(\mathbb R^d))
\]
notice first that the functional $v\mapsto F(\cdot,v(\cdot))$ mapping $C(\mathbb R^d)$ into itself is continuously differentiable and thus locally Lipschitz continuous. By \cite[Corollary~3.1.9]{Lunardi95}, the operator $\mathcal L$ generates an analytic semigroup in $C(\mathbb R^d)$. Hence, by \cite[Theorem~7.1.2]{Lunardi95} there exists a mild solution $v\in L^\infty(\tau,T-\delta; C(\mathbb R^d))$ to $-\partial_tv-\mathcal L v-F(y,v)=0$ for some $0\leq \tau <T-\delta$ when imposed at $t=T-\delta$ with a terminal value in $C(\mathbb R^d)$. Due to the a priori estimates established in Corollary~\ref{cor-a-priori-estimate} and \cite[Proposition~7.1.8]{Lunardi95} we may choose $\tau=0$. This gives us a global mild solution 
\[
	v\in L^\infty(0,T^-; C(\mathbb R^d))
\]
by pasting $v$ and $u$ at $T-\delta$. In order to verify the desired regularity we recall that $F$ is independent of the time variable. Hence, the regularity follows from \cite[Proposition~7.1.10(iv)]{Lunardi95} if $v(T-\delta,\cdot)\in D(\mathcal L)$ and if $\mathcal Lv(T-\delta,\cdot)+F(\cdot,v(T-\delta,\cdot))$ belongs to the real interpolation space $D_{\mathcal L}(\alpha,\infty)$. The former condition is a part of the assumption, the later is a consequence of \cite[Proposition~2.2.12(i)]{Lunardi95}. 
\end{proof}


\section{Verification argument} \label{sec-verification}

	This section is devoted to the verification argument. Throughout, for some $q>d+2$, let 
\[v\in W^{1,2}_{q,loc}((0,T^-)\times\mathbb R^d)\cap C_{poly}([0,T^-]\times\mathbb R^d)\] denote a nonnegative strong solution to~\eqref{pde-v}. We recall that by the parabolic Sobolev embedding theorems \cite[Lemma~II.3.3]{LSU68}, for every $R>0$ the parabolic Sobolev space $W^{1,2}_{q}((0,T)\times B_d(R))$ is continuous embedded into $C^{l/2,1+l}([0,T]\times\overline B_d(R))$ with $l= 1- (d+2)/q$.  Here, $B_{d}(R)$ denotes the $d$-dimensional ball of radius $R$ centered at the origin and $C^{l/2,1+l}([0,T]\times\overline B_d(R))$ denotes the usual parabolic H\"older space of functions $u(t,y)$ on $[0,T]\times\overline B_d(R)$ that are $l/2$-H\"older continuous in $t$ and  $l$-H\"older continuous in $y$ along with their first derivative in $y$. Hence, we assume from now on that $Dv$ is continuous.
	
	The verification argument is established as follows. We first prove that the candidate optimal strategy $(\xi^*,\pi^*)$ is admissible, and that the resulting portfolio process is monotone. This uses the lower estimate in~\eqref{a-priori-bounds}. Admissibility does not a priori guarantee that the strategy $(\xi^*,\pi^*)$ generates finite costs, though. This requires an extra argument.   
	
	Subsequently, we show that we may w.l.o.g.~restrict ourselves to  admissible controls that result in a monotone portfolio process. Similar to \cite{GraeweHorstQiu13,Kratz14}, we then prove the optimality of $\xi^*$ and $\pi^*$ in every interval $[t,s]$ with $s<T$. The upper estimate in~\eqref{a-priori-bounds} will be used to show that candidate strategy is optimal on the whole time interval.  
	
\begin{lemma}	\label{lemma-admissible}
	The pair of feedback controls $(\xi^*,\pi^*)$ given by \eqref{optimal-control} is admissible. The portfolio process $(X_s^*)_{s\in[t,T]}$ with respect to $(\xi^*,\pi^*)$ is monotone.
\end{lemma}
\begin{proof}
	One readily verifies that the portfolio process $(X_s^*)_{s\in[t,T]}$ with respect to the controls $\xi^*$ and $\pi^*$ is given by~\eqref{optimal-position-process} and thus is monotone. 
	To show that $X_T^*=0$ we define the random variable
\begin{equation*}
	\nu(\omega):=\sup\nolimits_{t\leq r\leq T}\left\{e^{\beta\theta(T-r)}\eta(Y_r^{t,y})^\beta E\left[\left.\sup\nolimits_{r\leq u\leq T}\eta(Y_u^{t,y})^{-\beta}\right|\mathcal F_r\right]\right\}
\end{equation*}
that is a.s.\ finite due to Assumption~\ref{A2} and the moment estimates \eqref{moment-estimate}. For $0\leq t\leq s<T$, using the lower estimate in~\eqref{a-priori-bounds} and Jensen's inequality we obtain,
\begin{equation*} 
	\begin{aligned}
	|X_s^*| &\leq|x|\exp\left(-\int_{t}^s\frac{v(r,Y_r^{t,y})^\beta}{\eta(Y_r^{t,y})^\beta}\,dr\right) \\
	&\leq|x| \exp\left(-\int_t^s\frac{1}{\eta(Y_r^{t,y})^\beta} E\left[\left.\frac{e^{-\theta(T-r)}}{\sqrt[\beta]{\int_r^T\eta(Y_u^{t,y})^{-\beta}\,du}}\right|\mathcal F_r\right]^\beta dr\right) \\
	&\leq|x| \exp\left(-\int_t^s\frac{e^{-\beta\theta(T-r)}\eta(Y_r^{t,y})^{-\beta}}{ E\left[\left.\int_r^T\eta(Y_u^{t,y})^{-\beta}\,du\right|\mathcal F_r\right]}\,dr\right) \\	
	&\leq |x| \exp\left(-\frac{1}{\nu}\int_t^s\frac{1}{T-r}\,dr\right)=|x|\left(\frac{T-s}{T-t}\right)^{1/\nu}\stackrel{s\rightarrow T}{\longrightarrow}0.
	\end{aligned}
\end{equation*}
This yields $X_{T-}^*=0$. Hence, $\pi_T^*=0$ and so $X_T^*=0$.
\end{proof}
	
	 	 
\begin{lemma} \label{lemma-suboptimal}
	For every $(\xi,\pi)\in\mathcal A(t,x)$ there exists $(\bar\xi,\bar\pi)\in\mathcal A(t,x)$ with lesser or equal costs such that $(X_s^{\bar\xi,\bar\pi})_{s\in[t,T]}$ is monotone.
\end{lemma}

\begin{proof}
	For $(\xi,\pi)\in\mathcal A(t,x)$ with $x\geq 0$ consider the strategy $(\bar\xi,\bar\pi)$ given by
\begin{equation*} 
	\bar\xi_s=\xi_{s} 1_{\{\xi_s\geq0\}} 1_{\{X_s^{\bar\xi,\bar\pi}>0\}}
\quad \text{ and } \quad	\bar\pi_s=(\pi_s\wedge X_{s-}^{\bar\xi,\bar\pi}) 1_{\{\pi_s\geq0\}} 1_{\{X_{s-}^{\bar\xi,\bar\pi}>0\}}.
\end{equation*}
By construction it holds that $0\leq\bar\xi_s\leq|\xi_s|$, $0\leq\bar\pi_s\leq|\pi_s|$ and $0\leq X_s^{\bar\xi,\bar\pi}\leq |X_s^{\xi,\pi}|$ for all $s\in[t,T]$. As a result, $(X_s^{\bar\xi,\bar\pi})_{s\in[t,T]}$ is monotone decreasing and $X_T^{\bar\xi,\bar\pi}=0$ by admissibility of $(\xi,\pi)$. Hence, $(\bar\xi,\bar\pi)\in\mathcal A(t,x)$ with less or equal costs than $(\xi,\pi)$. The case $x\leq 0$ is similar. 
\end{proof}

	We denote by $\bar{\mathcal A}(t,x)$ the set of all admissible controls under which the portfolio process is monotone. For any $(\xi,\pi)\in\bar{\mathcal A}(t,x)$ with finite costs the expected residual costs vanish as $s\rightarrow T$ as shown by the following lemma.	
\begin{lemma} \label{lemma-finite-costs}
	Under assumption (A3), for every $(\xi,\pi)\in\bar{\mathcal A}(t,x)$ with finite costs it holds that
\begin{equation} \label{end-costs-vanish}
	E\left[v(s,Y_s^{t,y})|X_s^{\xi,\pi}|^p\right] \longrightarrow 0, \quad \text{$s\rightarrow T$.}
\end{equation}
\end{lemma}

\begin{proof} 
	The monotonicity of $(X_s^{\xi,\pi})_{s\in[t,T]}$ together with the terminal condition $X_T^{\xi,\pi}=0$ implies
\begin{equation}  \label{alaba}
	|x|\geq |X_{s-}^{\xi,\pi}| \geq|X_s^{\xi,\pi}| \geq\left|\int_s^T\pi_r\,d N_r\right| \quad \text{and} \quad
	|X_{s-}^{\xi,\pi}|\geq|\pi_s|
\end{equation}
for all $t\leq s\leq T$, and moreover by Jensen's inequality
\begin{equation} \label{thiago}
	|X_s^{\xi,\pi}|^p\leq2^{p-1}\bigg(\left|\int_s^T\xi_r\,dr\right|^p+\left|\int_s^T\pi_r\,d N_r\right|^p\bigg).
\end{equation}
By It\^{o}'s formula,
\begin{equation*}
	\left|\int_s^T\pi_r\,d N_r\right|^p =\int_s^T\bigg\{\left|\int_r^T\pi_u\,d N_u+\pi_r\right|^p-\left|\int_r^T\pi_u\,d N_u\right|^p\bigg\}\,dN_r.
\end{equation*}
Using once more Jensen's inequality, and then \eqref{alaba}, we obtain
\begin{align}\label{lewi}
	\left|\int_s^T\pi_r\,d N_r\right|^p&\leq \int_s^T\bigg\{ (2^{p-1}-1)\left|\int_r^T\pi_u\,d N_u\right|^p+2^{p-1}\left|\pi_r\right|^p\bigg\}d N_r \nonumber\\
&\leq \int_s^T(2^{p}-1) |X_{r-}^{\xi,\pi}|^p\,d N_r.
\end{align}
As the integrand in \eqref{lewi} is bounded by $(2^p-1)|x|$ its jump integral has a true martingale part. Hence, by decomposing the jump integral in \eqref{lewi} into its martingale part with respect to the compensated Poisson process $N_t-\theta t$ and its deterministic part we obtain from \eqref{thiago},
\begin{equation*}
	|X_s^{\xi,\pi}|^p \leq 2^{p-1}E\left[\left.\left|\int_s^T\xi_r\,dr\right|^p+\int_s^T(2^{p}-1) |X_r^{\xi,\pi}|^p\,\theta dr\right|\mathcal F_s\right],
\end{equation*}
which implies by Gronwall's inequality the existence of a constant $C>0$ such that
\begin{equation*}
|X_s^{\xi,\pi}|^p\leq CE\left[\left.\left|\int_s^T\xi_r\,dr\right|^p\right|\mathcal F_s\right].
\end{equation*}
Next, we apply again Jensen's inequality to obtain,
\[
|X_s^{\xi,\pi}|^p\leq C (T-t)^{p-1}E\left[\left.\int_s^T|\xi_r|^p\,dr\right|\mathcal F_s\right].
\]
Therefore, by the upper estimate in \eqref{a-priori-bounds} and the boundedness of $\eta$ and $\lambda$ due to (A2),
\begin{align*} 
 E\left[v(s,Y_s^{t,y})|X_s^{\xi,\pi}|^p\right]&\leq C E\left[\frac{E\left[\left.\int_s^T\eta(Y_r^{t,y})+(T-r)^p\lambda(Y_r^{t,y})\,dr\right|\mathcal F_s\right]}{T-s}E\left[\left.\int_s^T|\xi_r|^p\,dr\right|\mathcal F_s\right]\right]\\
	&\leq \tilde C E\left[\int_s^T|\xi_r|^p\,dr\right].
\end{align*}
Letting $s\rightarrow T$, we conclude \eqref{end-costs-vanish} by the monotone convergence theorem, where it is used that $\xi\in L^p_{\mathcal F}(0,T;\mathbb R)$ for any strategy $(\xi,\pi)$ that has finite costs as $\eta$ is bounded away from zero under assumption (A2). 
\end{proof}


The following estimate is key to the verification argument. Together with the preceding lemma it allows us to show that $v(\cdot,\cdot)|\cdot|^p$ is indeed equal to the value function associated with our control problem. 

\begin{lemma} \label{generalized-ito}
	Under assumption (A1), for ever $(\xi,\pi)\in\bar{\mathcal A}(t,x)$ and $s\in[t,T)$ it holds,
\begin{equation*} 
	v(t,y)|x|^p \leq E\left[v(s,Y_s^{t,y})|X_s^{\xi,\pi}|^p\right]+ E\left[\int_t^sc(Y_r^{t,y},X_r^{\xi,\pi},\xi_r,\pi_r)\,dr \right].
\end{equation*}
\end{lemma}

\begin{proof}
		Let us denote by $B(y,R)$ the open ball with radius $R>0$ centered at $y \in \mathbb{R}^d$ and introduce the first exit time 
\[	
	\tau_R=\inf\{r\geq t: Y^{t,y}_r \notin B(y,R)\}. 
\]		
Since $v\in W^{1,2}_{q}((t,s)\times B(y,R))$ and $Y^{t,y}$ is non-degenerated, due to assumption~(A1), Krylov's generalized It\^o formula~\cite[Theorem~2.10.1]{Krylov80} applies to the stopped process $v(s\wedge\tau_R,Y_{s\wedge\tau_R}^{t,y})$. It yields that
\begin{equation*}
	v(t,y)= v(s\wedge\tau_R,Y_{s\wedge\tau_R}^{t,y})+ \int_t^{s\wedge\tau_R} \partial_t v(r,Y_r^{t,y})+\mathcal Lv(r,Y_r^{t,y})\,dr - \int_t^{s\wedge\tau_R}\sigma(Y_r^{t,y}) Dv(r,Y_r^{t,y})\,dW_r.
\end{equation*}
This allows us to apply to $v(s\wedge\tau_R,Y_{s\wedge\tau_R}^{t,y})|X_{s\wedge\tau_R}^{\xi,\pi}|^p$ the classical integration by parts formula for semimartingales \cite[Theorem~4.57]{JacodShiryaev03} in order to obtain
\begin{multline*}
	v(t,y)|x|^p = v(s\wedge\tau_R,Y_{s\wedge\tau_R}^{t,y})|X_{s\wedge\tau_R}^{\xi,\pi}|^p - \int_t^{s\wedge\tau_R}\big\{ \partial_tv(r,Y_r^{t,y})|X_r^{\xi,\pi}|^p+\mathcal Lv(r,Y_r^{t,y})|X_r^{\xi,\pi}|^p \\
	-p\xi_r v(r,Y_r^{t,y}) \sgn(X_r^{\xi,\pi})|X_r^{\xi,\pi}|^{p-1}+\theta v(r,Y_r^{t,y})(|X_r^{t,x}-\pi_r|^p-|X_r^{t,x}|^p)\big\}\,dr\\ 
	-\int_t^{s\wedge\tau_R}\sigma(Y_r^{t,y}) Dv(r,Y_r^{t,y})|X_r^{\xi,\pi}|^p\,dW_r  -\int_t^{s\wedge\tau_R} v(r,Y_r^{t,y})(|X_{r-}^{\xi,\pi}-\pi_r|^p-|X_{r-}^{\xi,\pi}|^p)\,d\widetilde N_r,
\end{multline*}
where $\widetilde N_r= N_r-\theta r$ denotes the compensated Poisson process. Both $v$ and $Dv$ are continuous and hence bounded on $[0,s]\times\overline B(y,R)$. Furthermore, $|X^{\xi,\pi}|\leq |x|$ and $|\pi|\leq |x|$, due to the monotonicity of the portfolio process. As a consequence, the above stochastic integrals are true martingales. Hence, recalling~\eqref{optimal-hamiltonian},
\begin{align}
	v(t,y)|x|^p 
	&=E\left[v(s\wedge\tau_R,Y_{s\wedge\tau_R}^{t,y})|X_{s\wedge\tau_R}^{\xi,\pi}|^p\right] +E\left[\int_t^{s\wedge\tau_R} c(Y_r^{t,y},X_r^{\xi,\pi},\xi_r,\pi_r)\,dr \right] \nonumber\\
	&\quad-E\left[\int_t^{s\wedge\tau_R} \big\{ (\partial_t+\mathcal L) v(r,Y_r^{t,y})|X_r^{\xi,\pi}|^p +H(r,Y_r^{t,x},X_r^{\xi,\pi},\xi_r,\pi_r,v(\cdot,\cdot)|\cdot|)\big\}\,dr \right]\nonumber\\
	&\leq  E\left[v(s\wedge\tau_R,Y_{s\wedge\tau_R}^{t,y})|X_{s\wedge\tau_R}^{\xi,\pi}|^p\right]  +E\left[\int_t^{s\wedge\tau_R} c(Y_r^{t,y},X_r^{\xi,\pi},\xi_r,\pi_r)\,dr \right] \label{suboptimal-estimate}\\
	&\quad- E\left[\int_t^{s\wedge\tau_R}\left\{ (\partial_t+\mathcal L) v(r,Y_r^{t,y})+F(Y_r^{t,y}, v(r,Y_r^{t,y}))\right\}|X_r^{\xi,\pi}|^p \,dr\right].\nonumber
\end{align}
Since $v$ satisfies \eqref{inflator-pde} a.e.,  since $Y^{t,y}$ is non-degenerated, and because $|X^{\xi,\pi}|\leq |x|$, if follows from Krylov's estimate~\cite[Theorem~2.4]{Krylov80} that
\[
E\left[\int_t^{s\wedge\tau_R}\left\{ (\partial_t+\mathcal L) v(r,Y_r^{t,y})+F(Y_r^{t,y}, v(r,Y_r^{t,y}))\right\}|X_r^{\xi,\pi}|^p \,dr\right]=0.
\]
Hence,
\[
	v(t,y)|x|^p \leq 
	E\left[v(s\wedge\tau_R,Y_{s\wedge\tau_R}^{t,y})|X_{s\wedge\tau_R}^{\xi,\pi}|^p\right] + 
	E\left[\int_t^{s\wedge\tau_R} c(Y_r^{t,y},X_r^{\xi,\pi},\xi_r,\pi_r)\,dr \right].
\]
Letting $R\rightarrow \infty$ the assertion follows from the polynomial growth condition on $v$ and positivity of the cost function $c$.
%
\end{proof}

	We are now ready to carry out the verification argument.

\begin{proof}[Proof of Proposition~\ref{thm-verifcation}]
%
Let $(\xi,\pi)\in\bar{\mathcal A}(t,x)$. 
By Lemma~\ref{generalized-ito} and Lemma~\ref{lemma-finite-costs} letting $s \to T$
(assuming w.l.o.g.\ that $(\xi,\pi)$ has finite costs) we get
\begin{equation*} 
 v(t,y)|x|^p\leq J(t,y,x;\xi,\pi).
\end{equation*}

	Finally note by Lemma~\ref{lemma-hjb} that equality holds in \eqref{suboptimal-estimate} if $\xi=\xi^*$ and $\pi=\pi^*$. Since $v$ and $c$ are both nonnegative this implies that 
\begin{equation} \label{optimal}
	v(t,y)|x|^p \geq \lim_{s \to T} E\left[\int_t^sc(Y_r^{t,y},X_r^{\xi^*,\pi^*},\xi_r^*,\pi_r^*)\,dr\right] = J(t,y,x;\xi^*,\pi^*).
\end{equation} 
In particular $(\xi^*,\pi^*)$ has finite costs. Hence, Lemma~\ref{lemma-finite-costs} applies to $(\xi^*,\pi^*)$. Thus,
\begin{equation*} \label{verification-in-expectation*}
\begin{split}	
	v(t,y)|x|^p &=  E[v(s,Y_s^{t,y})|X_s^{\xi^*,\pi^*}|^p]+ E\left[\int_t^sc(Y_r^{t,y},X_r^{\xi^*,\pi^*},\xi^*_r,\pi^*_r)\,dr\right] \\
	& \longrightarrow J(t,y,x;\xi^*,\pi^*)   \quad \text{as } s \to T.
\end{split}
\end{equation*}
This shows that the strategy $(\xi^*,\pi^*)$ is indeed optimal. 
%
\end{proof}

\section{Uniqueness in the non-Markovian framework} \label{sec-non-markov}

	Within our Markovian framework we obtained optimal controls in feedback form. Of course, one may as well interpret the cost coefficients as processes $\eta_t$, $\gamma_t$ and $\lambda_t$ adapted to the filtration generated by the Brownian motion. This has been recently suggested by Ankirchner, Jeanblanc \& Kruse~\cite{AnkirchnerJeanblancKruse14}, which allowed them to analyze non-Markovian coefficients, while losing the feedback form of the optimal controls. 
	
	Disregarding in this section any passive orders and assuming the filtration to be solely generated by the Brownian motion the value function to the control problem consider in~\cite{AnkirchnerJeanblancKruse14} is given by
\[
V_t(x):=\essinf_{\xi\in\mathcal A(t,x)}\Et{\int_t^T\eta_s|\xi_s|^p+\lambda_s|X_s^\xi|^p\,ds}, \qquad (t,x)\in[0,T)\times\mathbb R,
\]
where $\xi\in L_{\mathcal F}^0(t,T;\mathbb R)$ belongs to the set of admissible controls $\mathcal A(t,x)$ if the state process 
\[ 
	X_s^\xi=x-\int_t^s\xi_r\,dr, \qquad t\leq s\leq T,
\]
satisfies the liquidation constraint $X_T^\xi=0$. In the non-Markovian framework the value function has been related by Peng~\cite{Peng92} (see also \cite{GraeweHorstQiu13,HorstQiuZhang16}) to the BSPDE:
\begin{equation*}
-dV_t(x)=\inf_{\xi\in\mathbb R}\left\{-\xi\nabla V_t(x)+\eta_t|\xi|^p+\lambda_t|x|^p\right\} dt-\Psi_t(x)\,dW_t.
\end{equation*}
A solution to the BSPDE is a pair $(V_t,\Psi_t)$ of adapted processes. The ansatz $V_t(x)=Y_t|x|^p$ and $\Psi_t(x)=Z_t|x|^p$ results in the BSDE:
\begin{equation} \label{BSDE}
	-dY_t=\left\{\lambda_t-\frac{|Y_t|^{\beta+1}}{\beta\eta_t^\beta}\right\}dt-Z_t\,dW_t, \quad 0\leq t<T; \quad \lim_{t\rightarrow T} Y_t=+\infty.
\end{equation}
Assuming $\eta\in L^2_\mathcal F(0,T;\mathbb R_+)$, $\eta^{-\beta}\in L^1_{\mathcal F}(0,T;\mathbb R_+)$, $\lambda \in L^2_{\mathcal F}(0,T^-;\mathbb R_+)$, and $E[\int_0^T(T-t)^p\lambda_t\,dt]<\infty$ existence of a minimal nonnegative solution 
\[
	(Y,Z)\in L_{\mathcal F}^2(\Omega;C([0,T^-];\mathbb R_+))\times L_\mathcal F^2(0,T^-;\mathbb R^n)
\] to \eqref{BSDE} has been established in~\cite{AnkirchnerJeanblancKruse14}. 
We now show how our uniqueness argument can be applied to establish uniqueness in the BSDE setting when the coefficients are continuous. While the general shifting argument in the proof of our comparison principle fails in a non-time-homogeneous setting, it can be applied to establish the following a priori estimates.
\begin{proposition} \label{a-priori-bound}
	For any nonnegative solution $(Y,Z)$ to~\eqref{BSDE} the following estimates hold:
\begin{equation} \label{bsde-estimate}
	\frac{1}{\sqrt[\beta]{\Et{\int_t^T\frac{1}{\eta_s^\beta}\,ds}}}\leq Y_t\leq\frac{1}{(T-t)^p}\Et{\int_t^T\eta_s+(T-s)^p\lambda_s\,ds}=:\overline {Y_t}, \quad 0\leq t<T.
\end{equation}
\end{proposition}
\begin{proof} The lower estimate has already been established in \cite{AnkirchnerJeanblancKruse14} for the minimal nonnegative solution. To establish the upper estimate first note that $(\overline Y_t)_{t\in[0,T)}$ is a supersolution to~\eqref{BSDE}. Yet, one can not directly compare $\overline Y$ and $Y$ at the terminal time.	As a workaround, similarly as in \cite{Popier06}, we modify Pardoux's proof of his comparison principle for BSDEs with monotone drivers \cite[Theorem 2.4]{Pardoux99} by shifting the singularity of $\overline Y$. That is, for $\delta>0$ we define $(\overline{Y_t}^\delta)_{t\in[0,T-\delta)}$ by
\[\overline{ Y_t}^\delta=\frac{1}{(T-\delta-t)^p}\Et{\int_t^{T-\delta}\eta_s+(T-\delta-s)^p\lambda_s\,ds}.\]
These processes are again supersolutions to \eqref{BSDE} but with the singularity at $t=T-\delta$. Precisely, it holds that
\[-d\overline{Y_t}^\delta=\underbrace{\lambda_t+\frac{\eta_t}{(T-\delta-t)^p}-\frac{p\overline{Y_t}^\delta}{T-\delta-t}}_{=: g^\delta(t,\overline{Y_t}^\delta)}\,dt-\overline{Z_t}^\delta\,dW_t, \qquad 0\leq t<T-\delta,\]
for some $\overline Z^\delta\in\bigcap_{t\in[0,T-\delta)} L^2_{\mathcal F}(0,t;\mathbb R^n)$ given by the Martingale Representation Theorem with the singular terminal value \[\lim_{t\rightarrow T-\delta}\overline{Y_t}^\delta=+\infty.\]
A calculation as in~\eqref{tedious} verifies that for all $0\leq t<T-\delta$ and $y\in\mathbb R$,
\[g^\delta(t,y)\geq \lambda_t-\frac{|y|^{\beta+1}}{\beta\eta_t^\beta}=:f(t,y).\]

	We now consider the difference of $Y$ and $\overline Y^\delta$ for $0\leq t\leq s<T-\delta$:
\begin{align*}
\overline{Y_t}^\delta-Y_t&=\Et{\overline{Y_s}^\delta-Y_s+\int_t^s g^\delta(r,\overline{Y_r}^\delta)\,dr-\int_t^s f(r,Y_r)\,dr}\\
& =\Et{\overline{Y_s}^\delta-Y_s-\int_t^s\frac{p(\overline{Y_r}^\delta-Y_r)}{T-\delta-r}\,dr+\int_t^s g^\delta(r,Y_r)-f(r,Y_r)\,dr}.
\end{align*}
By the solution formula for linear BSDEs:
\begin{equation*}
\overline{Y_t}^\delta-Y_t=\Et{(\overline{Y_s}^\delta-Y_s)\exp\left(-\int_t^s\frac{p}{T-\delta-r}\,dr\right)+\int_t^s g^\delta(r,Y_r)-f(r,Y_r)\,dr}.
\end{equation*}
Therefore,
\begin{equation*}
\overline{Y_t}^\delta-Y_t\geq\Et{\bigg(\overline{Y_s}^\delta-\sup_{t\leq s\leq T-\delta} Y_s\bigg)\exp\left(-\int_t^s\frac{p}{T-\delta-r}\,dr\right)}.
\end{equation*}
Now, letting $s\rightarrow T-\delta$ this yields $\overline{Y_t}^\delta-Y_t\geq 0$ by Fatou's lemma. This completes the proof since
\[\overline{ Y_t}-Y_t=\lim_{\delta\rightarrow 0}\overline{Y_t}^\delta-Y_t\]
by the monotone convergence theorem.
\end{proof}

The upper estimate in \eqref{bsde-estimate} may be used to establish the analogue to Lemma~\ref{lemma-finite-costs}. Here we need to impose the following essentially boundedness assumption on the coefficients:
\begin{itemize}
	\item[(A5)] $\eta,\eta^{-1},(T-\cdot)^p\lambda\in L^\infty_{\mathcal F}(0,T;\mathbb R_+)$.
\end{itemize}
\begin{corollary}  \label{cor-remaining-costs}
Under assumption (A5), for any solution $(Y,Z)\in L_{\mathcal F}^2(\Omega;C([0,T^-];\mathbb R_+))\times L_\mathcal F^2(0,T^-;\mathbb R^n)$ to \eqref{BSDE} and any admissible control $\xi\in\mathcal A(t,x)$ with finite costs it holds that
\begin{equation} \label{vanishing-costs-non-markov}
	\lim_{s\rightarrow T} \Et{Y_s|X_s^\xi|^p}=0.
\end{equation}
\end{corollary}
\begin{proof}
	 From the upper estimate in~\eqref{bsde-estimate} and the tower property,
\[
	E\left[\left.Y_s |X_s^\xi|^p\right|\mathcal F_t\right]\leq E\left[\left.\frac{\int_s^T\{\eta_r+(T-r)^p\lambda_r\}\,dr}{(T-t)^p} |X_s^\xi |^p\right|\mathcal F_t\right].
\]
Taking the liquidation constraint into account, we obtain after an application of Jensen's inequality,
\begin{align*}
E\left[\left.Y_s |X_s^\xi|^p\right|\mathcal F_t\right]	&\leq  E\left[\left.\frac{\int_s^T\{\eta_r+(T-r)^p\lambda_r\}\,dr}{T-t}\int_s^T|\xi_r|^p\,dr\right|\mathcal F_t\right]\\
	&\leq C E\left[\left.\int_s^T|\xi_r|^p\,dr\right|\mathcal F_t\right] .
\end{align*}
Hence, letting $s\rightarrow T$, we conclude \eqref{vanishing-costs-non-markov} by the monotone convergence theorem, where it is used that $\xi\in L^p_{\mathcal F}(0,T;\mathbb R)$ for any control $\xi$ that has finite costs as $\eta$ is bounded away from zero under assumption (A5).
\end{proof}

   Finally, by a verification argument analogous to the one in Section~\ref{sec-verification} we obtain $V_t(x)=Y_t|x|^p$ for any nonnegative solution $(Y,Z)\in L_{\mathcal F}^2(\Omega;C([0,T^-];\mathbb R_+))\times L_\mathcal F^2(0,T^-;\mathbb R^n)$ and conclude: 

\begin{theorem}
  Under assumption (A5) uniqueness holds for problem~\eqref{BSDE} in the class of nonnegative solutions in $ L_{\mathcal F}^2(\Omega;C([0,T^-];\mathbb R_+))\times L_\mathcal F^2(0,T^-;\mathbb R^n)$.
\end{theorem}

\begin{proof}
We may again restrict the argument without loss of generality \cite[Lemma~1.6]{AnkirchnerJeanblancKruse14} to monotone controls $\xi\in\mathcal A(t,x)$ with finite costs. 
By It\^o-Kunita formula \cite[Theorem~I.8.1]{Kunita84}, since $Y_t|x|^p$ solves the stochastic HJB equation, 
\[
	Y_t|x|^p\leq\Et{Y_{s}|X_{s}^\xi|^p+\int_t^{s}\eta_r|\xi_r|^p+\lambda_r |X_r^\xi|^p\,dr}, \qquad t\leq s<T.
\]
Letting $s\rightarrow T$, we conclude $Y_t|x|^p\leq V_t(x)$ by Corollary~\ref{cor-remaining-costs}. Thus, $Y_t\leq V_t(1)$. But $V_t(1)$ is characterized in \cite{AnkirchnerJeanblancKruse14} as the minimal nonnegative solution to~\eqref{BSDE}. Hence, $Y_t=V_t(1)$ is unique.
\end{proof}


\section{Conclusion}

In this paper we proposed a novel approach to establishing smooth solutions to stochastic optimal control problems with singular terminal state constraints in a Markovian framework. Under standard assumptions on the diffusion and cost coefficients we proved that there exists at most one continuous viscosity and hence strong/classical solution to the HJB equation. As a byproduct we obtained a uniqueness theorem in a non-Markovian framework that complements results in \cite{AnkirchnerJeanblancKruse14}. Our main contribution is the existence of a classical solution under boundedness and differentiability assumptions on the coefficiants. Existence of a viscosity solution is still open.
In its present form our comparison principle only applies to continuous sub- and supersolutions, and hence does not allow us to apply Perron's method to establish the existence of a viscosity solution. Our verification argument uses Krylov's generalized It\^o formula. As such it applies to strong and classical solutions. It is not hard, though, to extend the verification argument to viscosity solutions. 


\appendix
\counterwithin{theorem}{section}
\section{Appendix}

In this appendix we present a modification of the comparison result for viscosity solutions given in \cite{BarlesBuckdahnPardoux97}. The original statement~\cite[Theorem~3.5]{BarlesBuckdahnPardoux97} concerns the uniqueness of viscosity solutions to systems of semilinear parabolic equations. The related comparison result is mentioned in \cite[Remark~3.9]{BarlesBuckdahnPardoux97}. As suggested in~\cite[Remark~6.105]{PardouxRascanu14}, the present scalar formulation covers the case of a monotone (not necessarily Lipschitz continuous) nonlinearity $G:[0,T]\times\mathbb R^d\times \mathbb R\rightarrow \mathbb R$. Namely, we assume 
\begin{itemize}
	\item[(A6)] $G$ is continuous,
	\item[(A7)] $(u-v)(G(t,y,u)-G(t,y,v))\leq \mu(u-v)^2$ for $\mu\in\mathbb R$ uniform in $t\in[0,T]$, $y\in\mathbb R^d$, $u,v\in\mathbb R$,
\end{itemize}
and consider for given terminal value $g:\mathbb R^d\rightarrow \mathbb R$ the following parabolic problem
\begin{equation} \label{generic-problem}
	\left\{\begin{aligned}&{-\partial_tu}(t,y)-\mathcal L u(t,y)-G(t,y,u(t,y))=0, \quad & (t,y)\in[0,T)\times\mathbb R^d,&\\
&u(T,y)=g(y), & y\in\mathbb R^d.&
\end{aligned}\right.
\end{equation}
It is worth mentioning that no uniform continuity type assumptions are needed due to the lack of any gradient dependence of $G$.
\begin{theorem} \label{theorem-finite-comparison}
	Let the assumptions (A6) and (A7) hold, and let $\underline u,\overline u\in C_{poly}([0,T]\times\mathbb R^d)$ be a viscosity sub- and a viscosity supersolution to~\eqref{generic-problem}, respectively. Then, $\underline u\leq\overline u$ in $[0,T]\times\mathbb R^d$.
\end{theorem}

	The proof is as in~\cite{BarlesBuckdahnPardoux97}. The only modification needed is to linearize the difference $G(t,y,\underline u(t,y))-G(t,y,\overline u(t,y))$ in terms of 
\[
	l(t,y):=1_{\underline u(t,y)\neq\overline u(t,y)}\frac{G(t,y,\underline u(t,y))-G(t,y,\overline u(t,y))}{\underline u(t,y)-\overline u(t,y)}.
\]	
rather than estimating it by a Lipschitz property, cf.~\cite[p.78]{BarlesBuckdahnPardoux97}. The key lemma \cite[Lemma~3.7]{BarlesBuckdahnPardoux97} based on the theorem of sums and the doubling variable technique then reads:

\begin{lemma} \label{lemma-linearization}
	The difference $w:=\underline u-\overline u$ is a viscosity subsolution to the linear equation
\begin{equation} \label{linearized-equation}
	-\partial_t w(t,y)-\mathcal Lw(t,y)-l(t,y)w(t,y)=0, \qquad (t,y)\in[0,T)\times\mathbb R^d.
\end{equation}
\end{lemma}

	Since $l$ is by (A7) bounded above by $\mu$ the supersolution property of the function given in~\cite[Lemma~3.8]{BarlesBuckdahnPardoux97} carries over to equation~\eqref{linearized-equation} so that the rest of the proof matches again with the original reference.

\bibliographystyle{siam}
{\small
\bibliography{bib_diss}
}

\end{document}